\documentclass[aps,prx,reprint]{revtex4-2}

\usepackage[utf8]{inputenc}

\usepackage{blindtext}
\usepackage{amsmath}
\usepackage{amsfonts}
\usepackage{amsthm}
\usepackage{physics}
\usepackage{makecell}
\usepackage{mathtools}
\usepackage[
    colorlinks=true,
    hypertexnames=false
    ]{hyperref}

\hypersetup{
  colorlinks   = true,
  urlcolor = magenta!90!black,
  linkcolor = blue!60!black,
  citecolor=black!60
}

\usepackage{url}
\usepackage[all]{hypcap}

\usepackage[capitalise,compress]{cleveref}
\usepackage{enumitem}   
\usepackage{tikz} 
\usetikzlibrary{arrows, shapes.gates.logic.US, calc}
\usetikzlibrary{shapes}
\usetikzlibrary{plotmarks}
\usetikzlibrary{quantikz2}
\usetikzlibrary{patterns}

\newcommand{\imag}{\mathrm{i}}

\DeclareMathOperator{\poly}{poly}
\DeclareMathOperator{\rec}{rec}
\DeclareMathOperator{\tra}{tra}
\DeclareMathOperator{\dup}{dup}
\DeclareMathOperator{\parity}{par}

\newcommand{\LES}{\mathrm{LES}}
\newcommand{\SIL}{\mathrm{SIL}}
\newcommand{\NSIL}{\mathrm{NSIL}}

\newcommand{\sort}{\mathrm{SORT}}
\newcommand{\revsort}{\mathrm{REVSORT}}
\newcommand{\unsort}{\mathrm{UNSORT}}
\newcommand{\shuffle}{\mathrm{SHUFFLE}}
\newcommand{\unshuffle}{\mathrm{UNSHUFFLE}}

\newcommand{\OT}{\tilde{O}}

\usepackage{xcolor}
\usepackage{bbm}

\newtheorem{theorem}{Theorem}

\newtheorem{definition}{Definition}
\newtheorem{result}{Result}

\crefname{protocol}{Protocol}{Protocols}
\Crefname{protocol}{Protocol}{Protocols}

\newlength{\protowidth}

\usepackage{algpseudocode}
\usepackage{algorithm}
\usepackage{changes}

\begin{document}

\title{Low-depth quantum symmetrization}

\author{Zhenning Liu, Andrew M. Childs, and Daniel Gottesman} 

\affiliation{Department of Computer Science, Institute for Advanced Computer Studies, and \\
Joint Center for Quantum Information and Computer Science, University of Maryland, College Park, MD 20742, USA
}

\begin{abstract}
Quantum symmetrization is the task of transforming a non-strictly increasing list of $n$ integers into an equal superposition of all permutations of the list (or more generally, performing this operation coherently on a superposition of such lists). This task plays a key role in initial state preparation for first-quantized simulations. Motivated by an application to fermionic systems, various algorithms have been proposed to solve a weaker version of symmetrization in which the input list is \emph{strictly} increasing, but the general symmetrization problem with repetitions in the input list has not been well studied. 
We present the first efficient quantum algorithms for the general symmetrization problem. If $m$ is the greatest possible value of the input list, our first algorithm symmetrizes any single classical input list using $\tilde{O}(\log n)$ depth and $O(n\log n + \log m)$ ancilla qubits, and our second algorithm symmetrizes an arbitrary superposition of input lists using $\tilde{O}(\log^3 n)$ depth and $O(n\log n)$ ancilla qubits. Our algorithms enable efficient simulation of bosonic quantum systems in first quantization and can prepare (superpositions of) Dicke states of any Hamming weight in $\tilde{O}(\log n)$ depth (respectively, $\tilde{O}(\log^3 n)$ depth) using $O(n\log n)$ ancilla qubits. We also propose an $\tilde{O}(\log^3 n)$-depth quantum algorithm to transform second-quantized states to first-quantized states. Using this algorithm, QFT-based quantum telescope arrays can image brighter photon sources, extending quantum interferometric imaging systems to a new regime.
\end{abstract}

\maketitle

\section{Introduction}\label{sec:introduction}
Symmetry is a central concept in physics. 
In particular, symmetry plays a key role in quantum simulation \cite{Feynman_1982}, where a faithful simulation of the many-body dynamics of indistinguishable particles should respect the symmetry of the simulated system. For \emph{fermions}, due to the Pauli exclusion principle, no two particles can occupy the same mode, and the wave function must be \emph{antisymmetrized} with respect to permutations of particles, leading to Fermi-Dirac statistics. For \emph{bosons}, multiple particles can be in the same state, and the wave function must be \emph{symmetrized} with respect to permutations, leading to Bose-Einstein statistics. A general quantum simulation procedure contains not only Hamiltonian evolution and final output, but also initial state preparation, which is often a challenging step, especially when the initial wave function has some specific structure or symmetry.

In this paper, we study a simple but subtle problem in quantum computing: transforming a list $l$ of \emph{non-strictly increasing} integers stored in a quantum register into the equal superposition of all permutations of $l$. We call this ``symmetrization," since the output state is symmetric with respect to any permutation of $l$. Symmetrization plays a fundamental role in quantum simulation based on \emph{first quantization}, an encoding of identical particles as a list of particle locations (whereas \emph{second quantization} represents the state as a list of occupation numbers of particle locations). Eigenstates in first quantization must be symmetric or antisymmetric if the system is bosonic or fermionic, respectively. Due to significant interest in quantum chemistry (which involves interactions between fermions), a weaker version of (anti)symmetrization where the list is \emph{strictly increasing} has been extensively studied \cite{abrams1997simulation,ward2009preparation,berry2018improved}. In particular, work of \textcite{berry2018improved} develops a logarithmic-depth algorithm. However, first-quantized simulation of bosons must allow multiple particles to be in the same mode, so the input list may contain repeated integers. Although several previous papers have claimed to handle this more general case \cite{abrams1997simulation,ward2009preparation,su2021fault,chan2023grid,kosugi2023first,georges2024quantum}, close examination reveals a subtlety that has not been properly addressed in the literature. The only solution for the general symmetrization problem that we are aware of is a result of \textcite{nepomechie2023qudit}. However, this algorithm requires prohibitively deep circuits: the depth is $O(n^m)$, where $n$ is the number of integers and $m$ is the greatest of these integers.

As the main results of this paper, we propose new logarithmic-depth quantum algorithms for symmetrization and a polylogarithmic-depth quantum algorithm for transforming from second to first quantization. We also present three applications of these algorithms. The first application, mentioned above, is \emph{initial state preparation in first-quantized simulation of bosons}, an unrecognized unsolved problem before our work. The second application is preparing a specific class of symmetric states useful in a variety of quantum information processing tasks, \emph{Dicke states}, where our algorithm is the first \footnote{We are aware of concurrent work that also gives an algorithm for preparing Dicke states with polylogarithmic depth but uses only polylogarithmically many ancillas \cite{jeffery}. The approaches are very different: while we use sorting networks, the other approach applies a simple sequence of collective rotations and parity measurements. Our approach solves a more general symmetrization problem, while the other approach is simpler, likely performs better in practice, and can perform better when the Hamming weight is lower.} to achieve logarithmic total circuit depth using a reasonable number of ancilla qubits. The third application is in optics, where we use our algorithm to improve strategies for long-baseline quantum telescope arrays \cite{gottesman2012longer,khabiboulline2019quantum} to process multiple photons at a time, yielding a \emph{quantum interferometric imaging system} that can take images in more general scenarios than previous designs.

In the remainder of this introductory section, we define the symmetrization problem formally in \cref{subsec:setup}. Then, we briefly state our results in \cref{subsec:ourresult}. Finally, in \cref{subsec:appli}, we introduce the three applications of symmetrization in more detail.

We present background for our results in \cref{sec:prelim}. More specifically, in \cref{subsec:sortingnetwork}, we introduce sorting networks and their quantization as the building blocks of symmetrization. Then in \cref{subsec:silsymmetrization}, we describe the algorithm of \textcite{berry2018improved}, which is used as a subroutine of one of our symmetrization algorithms. In \cref{subsec:reversible}, we generalize sorting to different comparison rules and define reversible sorting networks. In \cref{subsec:prefix}, we present the quantization of a standard prefix sum computation procedure.

In \cref{sec:nsilsymmetrize}, we discuss the general symmetrization problem with repetitions. We first explain in \cref{subsec:issues} why previous methods \cite{abrams1997simulation,ward2009preparation,berry2018improved} are not applicable to the more general symmetrization problem by discussing the subgroup structure of permutations when considering repetitions in the permuted list, and then extend quantum sorting networks to work with repetitions in \cref{subsec:qsn_nsil}. Using these theoretical tools, we present our new algorithm for symmetrizing a single input list in \cref{subsec:nsil_single}.

We present our symmetrization algorithm for an arbitrary superposition of input lists in \cref{subsec:nsil_superposed}. We present a novel SIL symmetrization procedure in \cref{appendix:newsil} based on quantization of a classical parallel algorithm from \textcite{alonso1996parallel}, which is subsequently used as a subroutine in \cref{appendix:exactnsil} to design the NSIL symmetrization algorithm for superposed inputs.

In \cref{sec:binaryencoder}, we propose a polylogarithmic-depth converter from a second quantized representation to a sorted list of modes occupied by the particles. This binary encoder, together with the symmetrization algorithm, gives the first polylogarithmic-depth conversion from second quantization to first quantization.

As a topic of independent interest, we demonstrate our design of a quantum interferometric imaging system in \cref{sec:inter}, including a pedagogical introduction of quantum telescope arrays that process one photon a time \cite{gottesman2012longer,khabiboulline2019quantum} in \cref{subsec:singlephoton}, followed by our extension to multiple-photon processing (using the conversion from second to first quantization) in \cref{subsec:multiplephoton}.

Finally, we summarize our results and discuss potential future works in \cref{sec:summary}.

\subsection{Problem setup}
\label{subsec:setup}
We design an $\OT(\log n)$-depth quantum algorithm for symmetrizing general increasing lists that may include repeated elements, which we call \emph{non-strictly increasing lists} (NSILs). Here we define $\OT(f(n)) \coloneqq O(f(n) \log(f(n)))$. Note that we will use this notation with $f(n)$ scaling like $\log n$, in which case the $\OT$ notation neglects a factor of $\log\log n$.

Recall that this symmetrization problem has been solved for \emph{strictly increasing lists} (SILs), which do not have any repetition. We formally define such lists as follows.

\begin{definition}[$\NSIL_n$ and $\SIL_n$]
For any nonnegative integer $n$,
    \begin{equation}
    \begin{aligned}
        \NSIL_n &:= \{ l\in\mathbb{Z}^n \mid \forall i<j, l_i\leq l_j\}\\
        \SIL_n &:= \{l\in\mathbb{Z}^n \mid \forall i<j, l_i < l_j\}
    \end{aligned}
    \end{equation}
are the sets of length-$n$ NSILs and SILs, respectively.

If the $l_i$s are restricted to be non-negative integers upper bounded by $m \in \mathbb{N}$, then we have
\begin{equation}
    \begin{aligned}
        \NSIL_n^m &:= \{l\in\{0,1,\dots,m\}^n \mid \forall i<j, l_i\leq l_j\}\\
        \SIL_n^m &:= \{l\in\{0,1,\dots,m\}^n \mid \forall i<j, l_i < l_j\}.
    \end{aligned}    
\end{equation}

\end{definition}

The input state for symmetrization is either $\ket{l}$ where $l$ is an NSIL of size $n$, or more generally, a superposition of such lists, i.e., 
\begin{equation}
    \sum_{l\in\NSIL_{n}} \alpha_l \ket{l}
\end{equation}
where $\sum_l |\alpha_l|^2 = 1$.

The output state of symmetrization for $\ket{l}$ is
\begin{equation}
\ket{\psi_\mathrm{out}(l)} \propto \sum_{\sigma \in S_n}  \ket{\sigma(l)}
\end{equation}
where $S_n$ is the group of permutations of $n$ elements. Similarly, for superposed inputs, the output state is proportional to
\begin{equation}
\sum_{l \in \NSIL_n} \alpha_l \ket{\psi_\mathrm{out}(l)}.
\end{equation}

\subsection{Our results}
\label{subsec:ourresult}
We provide two NSIL symmetrization algorithms, one for single NSIL inputs and another for superposed NSIL inputs. Both algorithms use SIL symmetrization algorithms as subroutines.

The first result uses a modified version of the \textcite{berry2018improved} algorithm and is designed for a single NSIL input. It takes $\OT(\log n)$ depth.

\begin{result}[$\OT(\log n)$-depth single-input NSIL symmetrization, \cref{subsec:nsil_single}]
\label{result:1}
For all $a\geq 2$, there exists an $\OT(\log n)$-depth algorithm using $O(a n \log n + \log m)$ ancilla qubits that produces $\ket{\psi_\mathrm{out}(l)}$ for any $l\in\NSIL_n$ with success probability $1-O(n^{-a+2})$.
\end{result}

Our second algorithm has slightly greater depth ($\OT(\log^3 n)$) but works for arbitrary superpositions of inputs $\sum_{l\in\NSIL_n}\alpha_l\ket{l}$ and always succeeds. It uses a novel SIL symmetrization algorithm as a subroutine, which is based on quantization of a classical parallel algorithm for generating random permutations \cite{alonso1996parallel}.

\begin{result}[$\OT(\log^3 n)$-depth superposed-input NSIL symmetrization, \cref{subsec:nsil_superposed}]
\label{result:2}
    There exists an $\OT(\log^3 n)$-depth NSIL symmetrization algorithm using $O(n \log n)$ ancilla qubits that produces $\sum_{l\in\NSIL_n}\alpha_l \ket{\psi_\mathrm{out}(l)}$ for any input state $\sum_{l\in\NSIL_n} \ket{\psi_\mathrm{out}(l)}$ with success probability $1$.
\end{result}

We also propose an $\OT(\log^2 n)$-depth algorithm to convert a list of occupation numbers to an $\NSIL$ of the occupied modes. Combining this algorithm and \cref{result:2}, we introduce the first polylogarithmic-depth quantum algorithm to transform a second-quantized state to a first-quantized state. Since this algorithm is a unitary, one can also transform from a first-quantized state to a second-quantized state by running it backward.

\begin{result}[$\OT(\log^3 n)$-depth conversion between second quantization and first quantization, \cref{sec:binaryencoder}]
\label{result:3}
    There exists an $\OT(\log^3 n)$-depth algorithm using $O(n \log n)$ ancilla qubits that transforms a state of $n$ particles in $m$ modes from a second-quantized representation to a first-quantized representation, or from a first-quantized representation to a second-quantized representation.
\end{result}

\subsection{Applications}
\label{subsec:appli}
\subsubsection{Bosonic eigenstate preparation}

For systems of $n$ indistinguishable particles and $m$ possible modes, second quantization uses $O(m\log n)$ qubits, while first quantization uses $O(n\log m)$ qubits, which is more favorable when $m\gg n$. However, for first quantization, the system symmetry must be reflected by the state, i.e., if two indistinguishable particles are exchanged, the state must be the same but with a $\pm 1$ phase depending on the particle type. Therefore, first-quantized eigenstates must be symmetrized or anti-symmetrized for bosons or fermions, respectively. Unlike for fermions, initial state preparation for first-quantized bosonic systems has not been extensively studied, even though it is a fundamental simulation primitive.

We review the history of initial state preparation for first-quantized bosonic systems. \textcite{abrams1997simulation} mainly considered Fermi systems and described a $\poly(n)$-depth symmetrization algorithm for SILs, which is improved by \textcite{ward2009preparation}. Although both papers claimed to be applicable for bosons, we explain why this is not the case in \cref{subsec:issues,subsec:qsn_nsil}. More recently, \textcite{berry2018improved} proposed the first $\OT(\log n)$-depth SIL symmetrization algorithm, as discussed in \cref{subsec:silsymmetrization,subsec:qsn_nsil}. Several papers proposed using these three algorithms to simulate bosons, without noticing the subtlety caused by repeated items. This includes the recent work of \textcite{su2021fault} which analyzed resource requirements for quantum chemistry simulation in first quantization and claimed that their algorithms also apply to bosons. Furthermore, \textcite{chan2023grid,georges2024quantum} discussed quantum chemistry simulation and suggested that initial state preparation for bosons can be done in a similar way as for fermions.

There are also several other recent theoretical developments in first-quantized simulation of bosons that do not discuss initial state preparation, including \textcite{berry2023quantum}, which discussed simulation of materials using the strategy in \textcite{su2021fault}; and \textcite{mukhopadhyay2024quantum}, which studies simulation of Pauli-Fierz Hamiltonians. Similarly, \textcite{tong2022provably} discussed simulation of gauge theories and bosonic systems, but the authors did not consider the computational cost of the initial state preparation.

As the most straightforward application of \cref{result:1,result:2}, the NSIL symmetrization algorithms enable efficient initial state preparation for first-quantized bosonic simulation. Furthermore, as mentioned in \cite{ward2009preparation}, one can also prepare a first-quantized initial state by preparing a second-quantized initial state and transforming it to its first-quantized counterpart. Our \cref{result:3} implements this transformation in poly-logarithmic depth. This idea is particularly useful if the second-quantized state is easier to prepare than the superposition of NSILs.

Although there are currently fewer proposals for first-quantized simulation of bosons than of fermions, our algorithm unlocks future possibilities in this direction if first quantization is shown to be more efficient in, for example, simulation of Bose-Hubbard models \cite{kuwahara2024effective} or nuclear effective theory \cite{watson2023quantum}.

\begin{table*}[t]
    \centering
    \begin{tabular}{|c| c|c| c| c|} 
     \hline \hline
     \textbf{Name} & \textbf{Depth} & \textbf{Repetitions} & \textbf{Ancillas} & \textbf{Connectivity} \\
     \hline
     \textcite{bartschi2019deterministic} & $O(n)$ & $1$& $0$ & LNN\\
     \hline 
     \textcite{bartschi2022short} & $O(n)$ & $1$ &$0$&LNN and all-to-all\\
     \hline
     \textcite{buhrman2023state} & $O(\log n)$ & $1$ & $\Omega(n^2)$ & LNN\\
     \hline
    \textcite{piroli2024approximating} & $O(1)$ & $O(\sqrt{n})$ & $O(n\log n)$ & LNN\\
     \hline
    \textcite{piroli2024approximating} & $O(\log n)$ & $O(\sqrt{n})$ & $O(n)$ & LNN\\
    \hline
     \cref{result:1} & $\OT(\log n)$ & $O(1)$ & $O(n\log n)$ & All-to-all\\
     \hline
     \cref{result:2} & $\OT(\log^3 n)$ & $1$ & $O(n\log n)$ & All-to-all\\
     \hline \hline
    \end{tabular}
    \caption{Comparison of Dicke state preparation algorithms for Hamming weight $k=O(n)$. Note that the algorithm of \textcite{piroli2024approximating} must be repeated $O(\sqrt{n})$ times, which implies either a larger number of ancilla qubits ($O(n^{3/2} \cdot \poly(\log n))$) or polynomial total circuit depth. 
    }
    \label{tab:dicke_comp}
\end{table*}

\subsubsection{Dicke state preparation}

The Dicke state $\ket{D_{n}^k}$ \cite{dicke1954coherence} is the equal superposition of all $n$-qubit computational-basis states with Hamming weight $k$, i.e.,
\begin{equation}
    \ket{D_n^k} := \frac{1}{\sqrt{\binom{n}{k}}} \sum_{r\in \{0,1\}^n, |r|=k} \ket{r}.
\end{equation}
Dicke states play an important role in various areas of quantum information science, including quantum networks \cite{prevedel2009experimental}, quantum metrology \cite{toth2012multipartite}, quantum error correction \cite{ouyang2014permutation}, and the quantum approximate optimization algorithm \cite{hadfield2019quantum}. Therefore, the quantum information community has sought circuits for preparing Dicke states with lower depth, gate complexity, and number of ancilla qubits. The most challenging Dicke states to prepare are those with $k$ close to $n/2$.
Linear-depth circuits for preparing these Dicke states were proposed in Refs.~\cite{bartschi2019deterministic,bartschi2022short}, but the first $\poly(\log n)$-depth circuits were proposed by \textcite{buhrman2023state}. These circuits use a model of ``local alternating quantum classical computations" in which the quantum operations performed are determined by the measurement outcomes of the previous step. However, the number of ancilla qubits required in that algorithm is huge ($\Omega(n^2)$) for large $n$. \textcite{piroli2024approximating} significantly reduce the number of ancilla qubits and circuit depth to $O(n\log n)$ and $O(1)$, respectively, or $O(n)$ and $O(\log n)$, respectively. However, the circuit in Ref.~\cite{piroli2024approximating} needs to be repeated $O(\sqrt{n})$ times in total, which implies either a larger number of ancilla qubits $O(n^{3/2} \cdot \poly(\log n))$ or polynomial total circuit depth.

Our low-depth NSIL symmetrization algorithm gives a straightforward strategy for preparing arbitrary Dicke states with $\OT(\log n)$ depth and $O(n\log n)$ ancilla qubits, with the cost of increased circuit connectivity. To prepare $\ket{D_n^k}$, it suffices to symmetrize
\begin{equation}
    \ket{\psi_{n,k}} := \ket{0}^{\otimes n-k} \otimes \ket{1}^{\otimes k}
\end{equation}
using our algorithm in \cref{subsec:nsil_single}, since the initial state is only one NSIL.

Superpositions of different Dicke states (with different $k$ values) are also useful in quantum metrology. For example, \textcite{lin2024covariantquantumerrorcorrectingcodes} show that metrological entanglement advantage can be achieved with such states. We can prepare such states in $\OT(\log^3 n)$ depth by first initializing the state as a superposition of different $\ket{\psi_{n,k}}$ states and then applying the full NSIL symmetrization algorithm described in \cref{subsec:nsil_superposed}. While some Dicke state preparation circuits can straightforwardly prepare superpositions of Dicke states\cite{bartschi2019deterministic,bartschi2022short}, this is not the case for previous polylogarithmic-depth circuits \cite{buhrman2023state,piroli2024approximating}, which may need to be modified and might incur extra overhead.

We summarize the comparisons of our algorithm with a selection of other Dicke state preparation circuits 
in \cref{tab:dicke_comp}.

\subsubsection{Quantum interferometric imaging}
As a promising candidate to achieve microarcsecond resolution in astronomical observation, long-baseline telescopes based on quantum repeaters were first proposed in \textcite{gottesman2012longer} and have sparked increasing research interest in both the astronomy and quantum information communities. In follow-up work, \textcite{khabiboulline2019quantum} design a \emph{quantum telescope array} and employ the quantum Fourier transform (QFT) to improve the signal-to-noise ratio. However, their method can only handle the case where the number of photons per mode is restricted to be much less than $1$ (i.e., the object must be extremely dim). Using our NSIL symmetrization algorithm, we extend the QFT-based scheme to cases where some modes have more than 1 photon. Since symmetrization is an intermediate step of interferometric imaging, we postpone the detailed description of this application to \cref{sec:inter}.

\section{Preliminaries}
\label{sec:prelim}

\subsection{(Quantum) sorting networks for SILs}
\label{subsec:sortingnetwork}
A sorting network is a circuit of \emph{comparators} that can be used to implement parallel comparison-based sorting algorithms. A classical comparator between elements $a_i$ and $a_j$ with $i< j$ swaps the values of $a_i$ and $a_j$ if $a_i> a_j$, and otherwise leaves them unchanged. For instance, a sorting network for 3 elements, corresponding to the well-known bubble sort algorithm, is shown in \cref{fig:bubble}.
Comparators can be quantized by adding an ancilla bit, as shown in \cref{fig:qcomp}. This construction implements a unitary comparator $U_\mathrm{comp}$ acting as
\begin{equation}
    U_\mathrm{comp} \ket{a,b,0} = \ket{\min(a,b),\max(a,b),\delta[a > b]}
    \label{eq:comp}
\end{equation}
where $\delta[a>b]$ is $1$ if $a>b$ and $0$ otherwise.

Note that if the integer is in $[m]$, then each $a$ or $b$ can be represented with $\log m$ bits. The comparator unitary for $\log m$ qubits can be implemented in $O(\log\log m)$ depth, as shown in Appendix B.2 of Ref.~\cite{berry2018improved}. When analyzing the circuit depth, we ignore this double-logarithmic overhead for simplicity.

Using the unitary comparator \eqref{eq:comp}, any sorting algorithm based on a sorting network of $N_a$ comparators can be directly quantized with $N_a$ ancilla qubits. Each ancilla qubit stores the result of the comparison of its corresponding comparator. In this paper, we combine all such ancilla qubits in a single register, and call its state after performing the sorting algorithm a \emph{record}. For $l\in \SIL_n$ and $\sigma\in S_n$, the record obtained when sorting $\sigma(l)$ is denoted $\rec(\sigma)$. Note that $\rec(\sigma)$ is independent of $l$ since for any $l,l' \in \SIL_n$, the record of sorting $\sigma(l)$ is the same as the record of sorting $\sigma(l')$. We write $\sort_{i,j}$ to represent the quantized sorting network that sorts a permuted SIL in register $i$ and stores the record in register $j$. Its operation is simply
\begin{equation}
    \sort_{i,j} \ket{\sigma(l)}_i\ket{0}_j = \ket{l}_i\ket{\rec(\sigma)}_j.
\end{equation}
As this is a linear operation, it can also be applied to a superposition of input states.

Apart from $\sort_{i,j}$, we make use of several other quantum operations based on quantum sorting networks. By running all gates of $\sort$ (which are either comparisons or cSWAPs) in the opposite direction, one can implement $\unsort_{i,j}$, the inverse of $\sort_{i,j}$, which acts as
\begin{equation}
    \unsort_{i,j} \ket{l}_i\ket{\rec(\sigma)}_j = \ket{\sigma(l)}_i\ket{0}_j.
\end{equation}
Furthermore, if we remove all comparison gates and only keep the cSWAP gates in $\unsort$, we obtain an operation denoted $\shuffle_{i,j}$, which applies the permutation stored in $\rec(\tau)$ to the input list:
\begin{equation}
    \shuffle_{i,j} \ket{\sigma(l)}_i\ket{\rec(\tau)}_j = \ket{\tau \circ \sigma(l)}_i\ket{\rec(\tau)}_j.
\end{equation}
Here the input list is not necessarily an SIL, so it is written as $\sigma(l)$. Since no comparison is performed in $\shuffle$, it has no dependence on the comparison rule.

Notice that $\sort$ essentially applies $\sigma^{-1}$ to $\sigma(l)$ to obtain $l$ using the information in the record generated by the comparison gates. If a record of $\tau$ is given and only the cSWAP gates in $\sort$ are performed, we get $\unshuffle$, the inverse operation of $\shuffle$:
\begin{equation}
    \unshuffle_{i,j} \ket{\sigma(l)}_i\ket{\rec(\tau)}_j = \ket{\tau^{-1}\circ \sigma (l)}_i\ket{\rec(\tau)}_j.
\end{equation}

Multiple sorting networks with sub-linear depth have been proposed, including the well-known bitonic sort \cite{batcher1968sorting} and merge sort \cite{stone1971parallel}, both featuring $O(\log^2 n)$ depth for a list of $n$ elements (and $N_a = O(n\log^2 n)$ ancilla qubits when quantized). As far as we are aware, the only known $O(\log n)$-depth sorting network is the AKS algorithm proposed in Ref.~\cite{ajtai19830}. However, this algorithm is impractical due to its large constant factor.

\begin{figure}[t]
\centering
\begin{quantikz}
\lstick{$a$} &\swap{1} &\qw  &\swap{1} &\qw\\
\lstick{$b$} &\targX{} &\swap{1} &\targX{} &\qw\\
\lstick{$c$} &\qw &\targX{} &\qw &\qw
\end{quantikz}
\caption{Bubble sort implemented by classical comparators.}
\label{fig:bubble}
\end{figure}
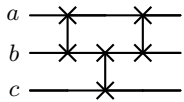

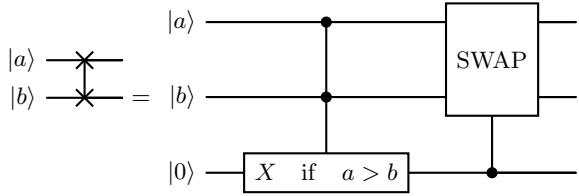
\begin{figure}[t]
\centering
\begin{quantikz}
\lstick{$\ket{a}$} &\swap{1} &\qw \\
\lstick{$\ket{b}$} &\targX{} &\qw \\
\end{quantikz}
=
\begin{quantikz}
\lstick{$\ket{a}$} &\ctrl{1} &\gate[2]{\mathrm{SWAP}} &\qw \\
\lstick{$\ket{b}$} &\ctrl{1} & &\qw\\
\lstick{$\ket{0}$} &\gate{X\quad \mathrm{if}\quad a> b} &\ctrl{-1} &\qw
\end{quantikz}
\caption{A quantum comparator consists of a comparison operator and a controlled-SWAP gate. An ancilla qubit is used to store the result of comparison and ensure the reversibility of the operation.}
\label{fig:qcomp}
\end{figure}

In this paper, we employ the AKS algorithm to design our symmetrization algorithms, but we emphasize that all sorting networks we use can be replaced by bitonic or merge sort if the number of elements to be sorted is not too large. While the asymptotic depth using bitonic or merge sort will have another factor of $\log n$, the performance may be better in practice. Since all four operations ($\sort,\unsort,\shuffle,\unshuffle$) only use as many gates as running a quantized sorting network, they can all be performed in $\OT(\log n)$ depth.

\subsection{Berry \emph{et al.} SIL symmetrization algorithms}
\label{subsec:silsymmetrization}
Symmetrization and antisymmetrization algorithms are fundamental building blocks of first-quantized quantum simulation for bosonic and fermionic systems. \textcite{abrams1997simulation} proposed the first quantum (anti-)symmetrization algorithm. This procedure has $\poly(n)$ depth. \textcite{berry2018improved} improved their algorithm by introducing $\log(n)$-depth sorting networks and a $\log(n)$-depth method for generating random permutations. However, both algorithms are only applicable to SIL inputs, although this was not explicitly pointed out in their papers. In this subsection, we briefly introduce the algorithm of \textcite{berry2018improved}, as it is a building block of our NSIL symmetrization algorithm.

Let $l\in \SIL_n^m$ be the input list. The algorithm uses 3 registers, and begins with the state
\begin{equation}
    \ket{l}_1 \ket{0\dots 0}_2 \ket{0\dots 0}_3.
\end{equation}
Register 2 is designed to hold a list of $n$ non-negative integers, where each integer is upper bounded by a pre-selected number $f(n)\gg n^2$. Therefore register 2 consists of $O(n\log n)$ qubits. Register 3 is the record register and also has $O(n\log n)$ qubits.

The algorithm starts by transforming the state in register 2 from $\ket{00\dots 0}$ into the product state
\begin{equation}
\begin{aligned}
\ket{\phi_2} &\propto \sum_{r_1=0}^{f(n)} \ket{r_1} \otimes \sum_{r_2=0}^{f(n)} \ket{r_2} \otimes \dots \otimes \sum_{r_n=0}^{f(n)} \ket{r_n}\\
&= \sum_{r_i\in[f(n)] ~ \forall i\in[n]} \ket{r_1 r_2 \dots r_n}.
\end{aligned}
\end{equation}
Write $r:=(r_1,r_2,\dots,r_n)$. Note that we can divide the $r$s into 2 categories:
\begin{itemize}
    \item $r$ is non-repetitive if $r_i\neq r_j ~\forall i\neq j$;
    \item $r$ is repetitive if $\exists i,j\in[n],\,i \ne j$ such that $r_i=r_j$.
\end{itemize}
The probability of $r$ being repetitive is at most $O(n^{-a+2})$ if each $r_i$ is uniformly randomly sampled from $[f(n)]$ and $f(n)\gg n^{a}$. Hence if we choose $a\geq 3$, then the probability of $r$ being repetitive is negligible. Moreover, if $r$ is non-repetitive, then there must be $\sigma\in S_n$ and $l'\in \SIL_n$ such that $r = \sigma(l')$. Therefore, we can rewrite $\ket{\phi_2}$ as
\begin{equation}
    \ket{\phi_2} \propto \sum_{l'\in \SIL_n^m} \sum_{\sigma \in S_n} \ket{\sigma(l')} + \ket{\perp}
\end{equation}
where $\ket{\perp}$ is the superposition of repetitive $r$s. Note that \textcite{berry2018improved} include a measure-and-postselect step to remove the repetitive part. We can also choose to use $a\geq 3$ and keep the repetitive part since it only leads to a negligible error.

Next, sort register 2 using the $\OT(\log n)$-depth sorting network and store the record in register 3:
\begin{equation}
\begin{aligned}
    &\ket{l}_1 \ket{\phi_2}_2 \ket{0\dots 0}_3 \\
    \stackrel{\sort_{2,3}}{\longrightarrow} & \ket{l}_1 \left( \sum_{l'\in \SIL_n^m} \ket{l'}_2  \sum_{\sigma\in S_n} \ket{\rec(\sigma)}_3  +\ket{\perp}_{23} \right).
\end{aligned}
\end{equation}
Now the 3 registers are almost independent. If register 2 is discarded, then with high probability ($1-O(n^{-a+2})$) the remaining part is $\ket{l}_1 \sum_{\sigma \in S_n} \ket{\rec(\sigma)}_3$ and the desired symmetrized state can be produced using $\unsort$:
\begin{equation}
 \ket{l}_1 \sum_{\sigma \in S_n} \ket{\rec(\sigma)}_3 \stackrel{\unsort_{1,3}}{\longrightarrow} \sum_{\sigma\in S_n} \ket{\sigma(l)}_1 \ket{00\dots 0}_3.
\end{equation}
In fact, the last step of the \textcite{abrams1997simulation} algorithm is also unsorting the input list using the superposition of all records in $S_n$, but the $\unsort$ operation is implemented based on $\poly(n)$-depth sorting algorithms.

\subsection{(Quantum) reversible sorting}
\label{subsec:reversible}
Even though sorting is an irreversible operation, we can define a notion of ``reversible sorting" as a transformation between lists that are promised to be sorted in different ways. First, we define a \emph{comparison rule} $\Gamma$ as a function
\begin{equation}
    \Gamma\colon \mathbb{R}^k \times \mathbb{R}^k \rightarrow \{-1,0,1\}
\end{equation}
satisfying
\begin{equation}
    \Gamma(x,y) = 0 \quad \mathrm{iff} \quad x=y
\end{equation}
and
\begin{equation}
    \Gamma(x,y) = 1 \quad \mathrm{AND} \quad \Gamma(y,z)=1 \implies \Gamma(x,z) = 1
\end{equation}
where $k\in \mathbb{Z}^+$. Every comparison rule defines a total order of numbers (or tuples of numbers). For instance, the ascending order of real numbers correspond to a comparison rule defined as $\Gamma(x,y)=1$ iff $x<y$ and $\Gamma(x,y)=-1$ iff $x>y$.
We can now define $\Gamma$-strictly-increasing lists (or $\Gamma$-sorted lists) as follows.

\begin{definition}[$\Gamma$-strictly-increasing lists ($\Gamma$SIL)]
    Suppose $l$ is a list of tuples of $k$ real numbers and $\Gamma$ is a comparison rule of $k$-tuples. Then $l$ is a $\Gamma$SIL if and only if $\Gamma(l_i,l_j) = 1$ for all $i<j$.
\end{definition}

We can also define comparison operators $<_\Gamma,=_\Gamma,>_\Gamma$ using this comparison rule:
\begin{equation}
    \begin{aligned}
        x<_{\Gamma}y \quad &\mathrm{iff} \quad \Gamma(x,y)=1,\\
    x =_{\Gamma} y \quad &\mathrm{iff} \quad \Gamma(x,y)=0,\\
    x >_{\Gamma} y \quad &\mathrm{iff} \quad \Gamma(x,y) = -1.
    \end{aligned}
\end{equation}

A $\Gamma$-sorting algorithm, denoted by $\mathcal{A}_{\Gamma}$, is a classical algorithm transforming a list to a $\Gamma$-sorted list. Such algorithms can be constructed for any $\Gamma$ using any classical comparison-based algorithm by replacing the usual $<$ operator with $<_\Gamma$. Now, suppose $\Gamma_1,\Gamma_2$ are comparison rules, and a list $l$ is $\Gamma_1$-sorted. Then sorting $l$ using $\mathcal{A}_{\Gamma_2}$ leads to a $\Gamma_2$-sorted list, denoted by $\mathcal{A}_{\Gamma_2}(l)$. This list can be sorted back to $l$ using $\mathcal{A}_{\Gamma_1}$, i.e., $\mathcal{A}_{\Gamma_1}(\mathcal{A}_{\Gamma_2}(l))=l$. Since the sorting process can be inverted if both orders are known, we call this process \emph{reversible} sorting.

An example of reversible sorting is as follows. Consider two comparison operators $<_{\Gamma_1},<_{\Gamma_2}$ for integers:
\begin{equation}
    x<_{\Gamma_1} y \quad \mathrm{iff} \quad x<y
\end{equation} and 
\begin{equation}
    \begin{aligned}
        x<_{\Gamma_2} y \quad \mathrm{iff} \quad (\parity(x)=1 ~ \mathrm{AND} ~ \parity(y)=0) \\
        \mathrm{OR} \quad (\parity(x)= \parity(y) ~ \mathrm{AND} ~ x<y)
    \end{aligned}
\end{equation}
where $\parity(x)$ is the parity of $x$. Letting $l_1=(1,2,3,4,5,7,8)$ and $l_2 = (1,3,5,7,2,4,8)$, one can observe that $l_1$ is $\Gamma_1$ sorted, $l_2$ is $\Gamma_2$ sorted, and we can transform between them using $\Gamma_1$- and $\Gamma_2$-sorting algorithms, respectively:
\begin{equation}
    \begin{aligned}
        &\mathcal{A}_{\Gamma_1} (1,3,5,7,2,4,8) = (1,2,3,4,5,7,8)\\
        &\mathcal{A}_{\Gamma_2} (1,2,3,4,5,7,8)= (1,3,5,7,2,4,8).
    \end{aligned}
\end{equation}

Observe that comparators can be generalized for any arbitrary comparison rule $\Gamma$: the unitary comparator $U_\mathrm{comp}^\Gamma$ acts as
\begin{equation}
    U_\mathrm{comp}^\Gamma \ket{a,b,0} = \ket{\min(a,b),\max(a,b),\delta[a >_{\Gamma} b]}
    \label{eq:compGamma}
\end{equation}
where $\delta[a>_{\Gamma}b]$ is $1$ if $a>_{\Gamma}b$ and $0$ otherwise. Replacing $U_\mathrm{comp}$ by $U_\mathrm{comp}^\Gamma$, one can also implement $\sort^\Gamma$ and $\unsort^\Gamma$ which sort and unsort tuple lists with respect to $\Gamma$, respectively.

Reversible sorting can be implemented using sorting networks. Reversibility implies that in principle we should be able to get rid of the record register, meaning that the reversible sorting operation $\revsort^{\Gamma_1,\Gamma_2}_i$ mapping a $\Gamma_1$-sorted list to a $\Gamma_2$-sorted list should perform
\begin{equation}
    \revsort^{\Gamma_1,\Gamma_2}_i\sum_{l_1 \in \Gamma_1\SIL} \alpha_{l_1}\ket{l_1} = \sum_{l_1\in \Gamma_1\SIL} \alpha_{l_1}\ket{\mathcal{A}_{\Gamma_2}(l_1)}.
\end{equation}

This can be implemented by uncomputing the record using sorting network operations with 3 ancilla registers, as follows:
\begin{equation}
    \begin{aligned}
    &\ket{\sigma(l)}_1\ket{0}_2\ket{0}_3\ket{12\dots n}_4\\
    \stackrel{\sort^{\Gamma_2}_{1,2}}{\longrightarrow} & \ket{l}_1 \ket{\rec(\sigma)}_2 \ket{0}_3 \ket{12\dots n}_4\\
    \stackrel{\sort^{\Gamma_1}_{1,3}}{\longrightarrow} & \ket{\sigma(l)}_1 \ket{\rec(\sigma)}_2 \ket{\rec(\sigma^{-1})}_3 \ket{12\dots n}_4\\
    \stackrel{\unsort_{4,2}}{\longrightarrow} & \ket{\sigma(l)}_1 \ket{0}_2 \ket{\rec(\sigma^{-1})}_3 \ket{\sigma(12\dots n)}_4\\
    \stackrel{\shuffle_{4,3}}{\longrightarrow}& \ket{\sigma(l)}_1 \ket{0}_2 \ket{\rec(\sigma^{-1})}_3 \ket{12\dots n}_4\\
    \stackrel{\unsort^{\Gamma_1}_{1,3}}{\longrightarrow}& \ket{l}_1 \ket{0}_2 \ket{0}_3 \ket{12\dots n}_4.
    \end{aligned}
\end{equation}
Note that $l$ is a $\Gamma_2$SIL and $\sigma(l)$ is a $\Gamma_1$SIL.

\subsection{(Quantum) parallel prefix sum computation}
\label{subsec:prefix}
In this subsection, we describe a standard parallel algorithm for computing the prefix sum of an array of $n$ integers, in a reversible manner. For a list of integers $(d_1,d_2,\dots,d_n)$, its prefix sum is
\begin{equation}
    (d_1, d_1+d_2, \dots, \sum_{j=1}^i d_j,\dots, \sum_{j=1}^n d_j).
\end{equation}
This algorithm, first proposed in \cite{hillis1986data}, is based on the divide-and-conquer paradigm and uses $O(\log^2 n)$ circuit depth.

This algorithm works on $n$ integer registers. The $i$th register, denoted by $d_i$, is initialized as the $i$th value of the input array. The algorithm contains $O(\log n)$ layers of subroutines, denoted by $D(h,t)$, where $h$ and $t$ (``head" and ``tail") are indices of $d$. Each $D(h,t)$ is applied after $D(h,(h+t)/2)$ and $D((h+t)/2,t)$.

We describe $D(h,t)$ in \cref{alg:difinding}.

\begin{algorithm}[H] 
	\caption{Subroutine $D(h,t)$}
        \label{alg:difinding}
	\begin{algorithmic}
    \If{$h<t$}
    \For{$i=(h+t)/2+1$ to $t$ \textbf{in parallel}}
        \State $d_i \leftarrow d_i + d_{(h+t)/2}$
    \EndFor
    \EndIf
	\end{algorithmic}
\end{algorithm}

Without loss of generality, we only consider the case where $n$ is a power of 2 (note that we can always add up to $O(n)$ zeros at the end of the input list to make the total number a power of 2). Now, the prefix sums can be computed classically using $\log n$ layers of $D(h,t)$, as shown in \cref{alg:prefixsum}.
\begin{algorithm}[H] 
	\caption{Parallel prefix sum computation}
        \label{alg:prefixsum}
	\begin{algorithmic}
    \For{$i=1$ to $\log n$}
        \For{$j=0$ to $n / 2^i - 1$ \textbf{in parallel}}
            \State $D(j\cdot 2^i + 1,(j+1)\cdot 2^i)$
        \EndFor
    \EndFor
    \end{algorithmic}
\end{algorithm}

The subroutine $D(h,t)$ can be implemented in $O(\log n)$ depth since it involves adding $d_{(h+t)/2}$ to $(t-h)/2$ registers. Hence the total circuit depth is $\OT(\log^2 n)$.

It is straightforward to quantize the prefix sum computation by quantizing each $D(h,t)$. An example is shown in \cref{fig:quantization_d}. The quantum circuit depth is also $O(\log^2 n)$.

\begin{figure}
    \centering
    \centering
            \begin{quantikz}
                \lstick{$\ket{d_1}$} & \gate[8]{U_D} &\qw \rstick{$\ket{d_1}$} \\
                \lstick{$\ket{d_2}$} & &\qw \rstick{$\ket{d_2}$}\\
                \lstick{$\ket{d_3}$} & &\qw \rstick{$\ket{d_3}$}\\
                \lstick{$\ket{d_4}$} & &\qw \rstick{$\ket{d_4}$}\\
                \lstick{$\ket{d_5}$} & &\qw \rstick{$\ket{d_5+d_4}$}\\
                \lstick{$\ket{d_6}$} & &\qw \rstick{$\ket{d_6+d_4}$}\\
                \lstick{$\ket{d_7}$} & &\qw \rstick{$\ket{d_7+d_4}$}\\
                \lstick{$\ket{d_8}$} & &\qw \rstick{$\ket{d_8+d_4}$}
            \end{quantikz}
    \caption{Quantization of $D(1,8)$. This unitary can be implemented in $O(\log n)$ depth.}
    \label{fig:quantization_d}
\end{figure}
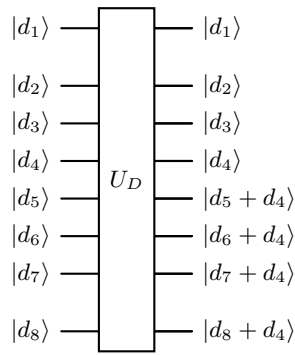

\section{Symmetrization of non-strictly increasing lists}
\label{sec:nsilsymmetrize}

In this section, we consider symmetrizing non-strictly increasing that may contain duplicated items, i.e., $\NSIL$s. In \cref{subsec:issues}, we introduce some useful subsets and subgroups of $S_n$ related to $\NSIL$s. In \cref{subsec:qsn_nsil}, we describe the behaviors of $\sort$ and $\unsort$ on NSILs, and explain why the SIL symmetrization algorithms do not apply to $\NSIL$s. In \cref{subsec:nsil_single}, we describe our new algorithm for symmetrizing a single $l$ where $l\in\NSIL$. We defer the algorithm for superpositions of $\NSIL$s (i.e., inputs of the form $\sum_l \alpha_l \ket{l}$) to \cref{subsec:nsil_superposed}.

\subsection{Degenerate permutations for NSILs}
\label{subsec:issues}
There is a kind of degeneracy of permutations for $l\in[m]^n$ with repetitions (in this paper, we mostly study $l\in\NSIL_n^m \subset [m]^n$): different permutations $\sigma_1, \sigma_2 \in S_n$ may lead to the same permuted list, i.e., $\sigma_1(l)=\sigma_2(l)$. In this subsection, we show that this degeneracy is the fundamental reason why the \textcite{berry2018improved} algorithm does not apply to NSILs. In particular, let $H_l$ denote a subset of the permutation group $S_n$ which contains only permutations preserving the input $l\in [m]^n$, i.e.,
\begin{equation}
    H_l := \{h:h(l)=l\}.
\end{equation}
Then $\sigma \circ h (l) = \sigma (l)$ for all $\sigma\in S_n, h\in H_l$. Note that all $h\in H_l$ must only permute duplicated elements in $l$. For instance, if $l=12334$, then $H_l$ contains only the identity permutation and the swap between the $3^\mathrm{rd}$ and the $4^\mathrm{th}$ elements. One can show that $H_l$ is a subgroup of $S_n$ since $h_1 \circ h_2(l) = l$ if $h_1,h_2\in H_l$. Furthermore, every permutation $\sigma\in S_n$ can be uniquely decomposed as $\sigma = \sigma' \circ h$ where $h\in H_l$ and $\sigma'$ is a permutation that preserves the relative order of duplicated items. For the same example of $l=12334$ where elements 3 and 4 are duplicated, consider a permutation $\sigma$ satisfying $\sigma(12345)=51423$. The relative position of 3 and 4 is exchanged in $\sigma$, which can be uniquely extracted as another permutation $h$ satisfying $h(12345)=12435$. Now, the fact that 3 and 4 are out of order in $\sigma$ can be ``corrected", and we can write $\sigma = \sigma' \circ h$ where $\sigma'(12345)=51324$. In fact, all elements of a left coset of $H_l$ in $S_n$ permute $l$ in the same way:
\begin{equation}
    \sigma H_l (l) = \{\sigma(l)\}.
\end{equation}
Since every element of $S_n$ must belong to exactly one left coset of $H$ in $S_n$, we conclude that for any $l\in[m]^n$, 
\begin{equation}
    S_n = R_lH_l := \{\sigma' \circ h: \sigma'\in R_l, h\in H_l\}, 
\end{equation}
where $R_l$ denotes the set of all permutations preserving the relative orders of duplicated elements in the list $l$. Note that $R_l$ is a set of canonical elements of all left cosets of $H$ in $S_n$, which can be chosen arbitrarily.

This leads to another observation: a superposition of all permutations in $S_n$ acting on $l$ and $123\dots n$, respectively, can be represented as
\begin{equation}
\begin{aligned}
    &\sum_{\sigma\in S_n} \ket{\sigma(l)} \ket{\sigma(12\dots n)}  \\
    &\quad = \sum_{\sigma' \in R_l} \ket{\sigma'(l)} \sum_{h\in H_l} \ket{\sigma'\circ h (12\dots n)}.
\end{aligned}
\end{equation}

\subsection{Quantum sorting networks for NSILs}
\label{subsec:qsn_nsil}
Recall that when introducing sorting-network-related operations in \cref{sec:prelim}, $\sort$ is designed only to sort $\sigma(l)$ where $l$ is strictly increasing, and $\unsort$ is also only designed to act on SILs. Here we discuss what happens if both operations are applied to NSILs. When $l$ contains duplicates, the $\sort$ operation still transforms $\sigma(l)$ to $l$, but what is stored in the record depends on the \emph{stability} of the sorting algorithm. Recall from \cref{sec:prelim} that we are free to choose the sorting algorithm. A sorting algorithm is called \emph{stable} if it preserves the relative order of duplicated elements. Therefore, sorting any permutation in a given coset using a stable algorithm produces the record corresponding to the canonical element we picked ($R_l$), i.e., the record of sorting $\sigma(l)$ will be $\rec(\sigma')$ where $\sigma'(l) = \sigma(l)$ and $\sigma'\in R_l$. 
(Recall that $\rec(\sigma)$ is the record obtained from sorting $\sigma(l_\mathrm{S})$ where $l_\mathrm{S}$ can be any $\SIL$.)
On the other hand, unstable sorting networks shuffle the relative order of duplicated elements, so the record of sorting $\sigma(l)$ will be $\rec(\sigma' \circ h')$ where $h'$ is an unknown permutation in $H_l$ depending on $\sigma'$. In this paper, we consider general unstable sorting networks, since many efficient sorting networks, including bitonic sort \cite{ajtai19830}, are unstable.

The effect of $\unsort$ is less clear. We consider two cases:
\begin{enumerate}
    \item If the record is generated by performing $\sort$ on $\ket{\sigma'(l)}\ket{00\dots 0}$, then the record is $\ket{\rec(\sigma'\circ h')}$ where $\sigma'\in R_l, h'\in H_l$. The operation of $\unsort$ is
    \begin{equation}
    \unsort_{1,2} \ket{l}_1 \ket{\rec(\sigma'\circ h')}_2 = \ket{\sigma'(l)}_1 \ket{00\dots 0}_2
    \end{equation}
    since $\unsort$ is the inverse of $\sort$.
    \item If the record is $\rec(\sigma'\circ h_\mathrm{S})$ with $\sigma'\in R_l, h_\mathrm{S}\in H_l$ and $h_\mathrm{S} \neq h'$ (as can be generated by sorting $\sigma' \circ h_\mathrm{S} (l_\mathrm{S})$ where $l_\mathrm{S} \in \SIL$), then $\sigma'(l)$ can still be generated, but the record register is \emph{not} reset to $\ket{00\dots}$. Instead, $\unsort$ acts as
    \begin{equation}
        \unsort_{1,2} \ket{l}_1 \ket{\rec(\sigma'\circ h_\mathrm{S})}_2 =\ket{\sigma'(l)}_1 \ket{\tra(\sigma',h_\mathrm{S},l)}_2
    \end{equation}
    where $\ket{\tra(\sigma',h_\mathrm{S},l)}$ is a \emph{trash state} depending on the choice of $\sigma', h_\mathrm{S}$, and $l$,  and must be orthogonal to $\ket{00\dots 0}$.
    The is because $\unsort$ is a unitary and $\langle \rec(\sigma'\circ h') | \rec(\sigma'\circ h_\mathrm{S})\rangle = 0 $.
    The trash state contains unremovable $1$s in the record that are generated by swaps between inversions of $i$ and $j$ in $\sigma'\circ h_\mathrm{s} (12\dots n)$ with $l_i = l_j$.
\end{enumerate}

To illustrate the problem, we give a very simple example for the bubble sort network shown in \cref{fig:bubble}. If the record is generated by sorting $\sigma(l_\mathrm{S})=312$, then $\rec(\sigma) = 110$. (In the first comparator, $3>1$, hence its record is $1$; in the second comparator, $3>2$, hence its record is $1$; in the last comparator, $1<2$, hence its record is $0$.) Now, if $l=122$, unsorting $122$ using the record leads to the following transcript: in the last comparator, no swap is performed and $1<2$, hence its record is unchanged; in the second comparator, a swap is performed due to the record, but since $2=2$, the record is \emph{unchanged} and still $1$; in the first comparator, a swap is performed and $2>1$, so the record is flipped to $0$. The state of the record after unsorting is $010$, instead of the desired $000$. In \cref{tab:trashstates}, we give a complete list of bubble-sort-based $\unsort$ operations on $122$ for all possible records generated by sorting $\sigma(123)$.

Recall that the common final step of existing SIL symmetrization algorithms is unsorting the input $l\in \SIL$ using the superposition of all records of $\sigma\in S_n$. The observations above indicate that, if instead $l\in \NSIL$, these algorithms would in general not work, because the scratch qubits in register 2 are entangled with register 1, such that one cannot obtain the desired superposition by discarding them.

Let us consider what kind of sorting algorithms can make $\unsort$ leave registers 1 and 2 independent, so the state in register 1 is the desired superposition. This can only be achieved if the set of trash states obtained by sorting the members of a left coset of $H_l$ in $S_n$ (such as $\sigma' H_l$ with $\sigma'\in R_l$) is the same for all those left cosets, that is to say,
\begin{equation}
    \{\tra(\sigma',h,l) \mid \forall h\in H_l \}
\end{equation}
is the same set for all $\sigma'\in R_l$. In this case, register 1 is in the superposition of permuted lists over all left cosets $\sigma'H_l$, and register 2 is in the superposition of all possible trash states that are independent of the choice of coset in register 1. Although we cannot rule out the possibility that such sorting algorithms exist, it does not seem straightforward to find them. We thus leave this as an open problem.

Moreover, even if such a sorting algorithm exists, the corresponding SIL symmetrization algorithm will still fail when the input list is a superposition of different NSILs with different patterns of repetition, as the trash states also depend on the choice of $l$.

\begin{table}[]
    \centering
    \begin{tabular}{|c|c|c|c||c|}
        \hline
         $\sigma(123)$ & $\rec(\sigma)$ & $\sigma(122)$ & $\ket{\tra}$ & Action of $\unsort$ \\
         \hline
         \hline
         123 & 000 & 122 & 000 & $\ket{122}\ket{000} \mapsto \ket{122}\ket{000}$ \\
         \hline
         132 & 010 & 122 & 010 & $\ket{122}\ket{010} \mapsto \ket{122}\ket{010}$ \\
         \hline
         213 & 100 & 212 & 000 & $\ket{122}\ket{100} \mapsto \ket{212}\ket{000}$\\
         \hline
         231 & 011 & 221 & 000 & $\ket{122}\ket{011} \mapsto \ket{221}\ket{000}$\\
         \hline
         312 & 110 & 212 & 010 & $\ket{122}\ket{110} \mapsto \ket{212}\ket{010}$\\
         \hline
         321 & 111 & 221 & 100 & $\ket{122}\ket{111} \mapsto \ket{221}\ket{100}$\\
         \hline
    \end{tabular}
    \caption{The full table of the action of $\unsort$ (based on bubble sort) on $122$ when the records of all 6 possible permutations of $123$ are given.}
    \label{tab:trashstates}
\end{table}

\subsection{NSIL symmetrization for a single input}
\label{subsec:nsil_single}
In this subsection, we describe our algorithm for symmetrizing NSILs for input state $\ket{l}$ using our observations mentioned in \cref{subsec:issues,subsec:qsn_nsil}. The algorithm extensively uses all four sorting network operations introduced in \cref{sec:prelim}. We use three registers (the latter two are ancilla registers): register 1 stores the input NSIL $l$, register 2 stores all records for the sorting algorithm and is initialized as all zero, and register 3 stores the superposition of all $\sigma(12\dots n)$ initially, which can be generated using either the \textcite{berry2018improved} SIL symmetrization algorithm (with $\OT(\log n)$ depth) or our new algorithm described in \cref{appendix:newsil} (with $\OT(\log^3(n))$ depth). Register 3 is used as a platform to mediate multiplications of permutations, which eventually allows us to extract the subgroup structure corresponding to NSILs.

As the first step, we sort each $\sigma(12\dots n)$ reversibly using sorting networks and store the records of swaps:
\begin{equation}
\begin{aligned}
    \ket{l}_1 \ket{00\dots 0}_2 \sum_{\sigma \in S_n} \ket{\sigma(12\dots n)}_3\\
  \stackrel{\sort_{3,2}}{\longrightarrow} \ket{l}_1 \sum_{\sigma \in S_n} \ket{\rec(\sigma)}_2 \ket{12\dots n}_3.
\end{aligned}
\end{equation}

Next, we apply the swaps stored in the record to shuffle register 1, and then unsort register 3:
\begin{equation}
\begin{aligned}
    \stackrel{\shuffle_{1,2}}{\longrightarrow} 
\sum_{\sigma \in S_n} \ket{\sigma(l)}_1 \ket{\rec(\sigma)}_2 \ket{12\dots n}_3\\
\stackrel{\unsort_{3,2}}{\longrightarrow} \sum_{\sigma \in S_n} \ket{\sigma(l)}_1 \ket{00\dots 0}_2 \ket{\sigma(12 \dots n)}_3.
\end{aligned}
\end{equation}

The decomposition of a permutation for non-strictly increasing lists can now be applied to rewrite the above state:
\begin{equation}
\label{eq:rewritten}
\begin{aligned}
    &\sum_{\sigma \in S_n} \ket{\sigma(l)}_1 \ket{00\dots 0}_2 \ket{\sigma(12\dots n)}_3 \\
    &=  \sum_{\sigma' \in R_l} \ket{\sigma'(l)}_1 \ket{00\dots 0}_2 \sum_{h\in H_l}  \ket{\sigma'\circ h(12\dots n)}_3.
\end{aligned}
\end{equation}

Now, we perform sorting network operations on the state above to sort the $\sigma'(l)$ strings:

\begin{equation}
\begin{aligned}
&\sum_{\sigma' \in R_l} \ket{\sigma'(l)}_1 \ket{00\dots 0}_2 \sum_{h\in H_l}  \ket{\sigma'\circ h(12\dots n)}_3\\
    \stackrel{\sort_{1,2}}{\longrightarrow}& \sum_{\sigma' \in R_l} \ket{l}_1 \ket{\rec(\sigma'\circ h')}_2 \sum_{h\in H_l} \ket{\sigma'\circ h(12\dots n)}_3.
\end{aligned}
\end{equation}
Next, we apply the swaps recorded in register 2 to unshuffle register 3:
\begin{equation}
\begin{aligned}
    \stackrel{\unshuffle_{3,2}}{\longrightarrow} & \sum_{\sigma' \in R_l} \ket{l}_1 \ket{\rec(\sigma'\circ h')}_2 \cdot \\
    &\quad\cdot \sum_{h\in H_l} \ket{h'^{-1}\circ \sigma'^{-1}\circ \sigma'\circ h(12\dots n)}_3\\
    &=\sum_{\sigma' \in R_l} \ket{l}_1 \ket{\rec(\sigma'\circ h')}_2 \cdot \\
    &\quad\cdot \sum_{h\in H_l} \ket{h'^{-1}\circ h(12\dots n)}_3\\
    &=\sum_{\sigma' \in R_l} \ket{l}_1 \ket{\rec(\sigma'\circ h')}_2 \sum_{h\in H_l} \ket{h(12\dots n)}_3
\end{aligned}
\end{equation}
where in the last line we used the fact that $h'^{-1}\in H_l$ and $h'^{-1}H_l = H_l$ when $h'\in H_l$.

Now, to obtain a superposition of $\sigma(l)$ strings, it suffices to unsort $l$ using the records:
\begin{equation}
\label{eq:lastunsort}
    \stackrel{\unsort_{1,2}}{\longrightarrow} \sum_{\sigma' \in R_l} \ket{\sigma'(l)}_1 \ket{00\dots 0}_2 \sum_{h\in H_l} \ket{h(12\dots n)}_3.
\end{equation}
We notice that the strings in register 3 are now only superpositions of $h(12\dots n)$ over $h\in H_l$, and are independent of which permutation is applied in register 1. Therefore, if the input is a single NSIL basis state $\ket{l}$, one can simply discard registers 2 and 3, and the state in register 1 is $\sum_{\sigma\in S_n} \ket{\sigma(l)}$. The output fidelity of single-input symmetrization is $1$ if the resource state $\sum_\sigma \ket{\sigma(12\dots n)}$ is prepared exactly. The \textcite{berry2018improved} algorithm can produce this state with high probability, as mentioned in \cref{subsec:silsymmetrization}.

\section{NSIL symmetrization for a superposed input}
\label{subsec:nsil_superposed}

If the input to the above procedure is some arbitrary superposition of different NSILs with different patterns of repetitions, registers 1 and 3 will be entangled. As a simple example, if $\ket{\psi_\mathrm{in}} = \frac{1}{\sqrt{3}} \ket{1223} + \sqrt{\frac{2}{3}}\ket{1333}$, then the state of the algorithm at the step of Equation~\eqref{eq:lastunsort} is
\begin{equation}
\begin{aligned}
    &\frac{1}{\sqrt{3}} \sum_{\sigma' \in S_{1223}} \ket{\sigma'(1223)}_1 \ket{00\dots 0}_2 \sum_{h\in H_{1223}} \ket{h(1234)}_3\\
    + &\sqrt{\frac{2}{3}} \sum_{\sigma' \in S_{1333}} \ket{\sigma'(1333)}_1 \ket{00\dots 0}_2 \sum_{h\in H_{1333}} \ket{h(1234)}_3,
\end{aligned}
\end{equation}
from which we cannot obtain $\ket{\psi_\mathrm{out}}$ by simply discarding register $3$. Hence, we must erase the scratch qubits in register 3 for all possible input strings $l$ coherently. This can be done by inverting the circuit preparing the superposition of the $H_l$ subgroup ($\sum_{h\in H_l} \ket{h(12\dots n}$) from $\ket{l}$. Therefore, in this section, we design a unitary to implement the following subgroup superposition preparation procedure:
\begin{equation}
    \sum_l \alpha_l \ket{l} \ket{12\dots n} \mapsto \sum_l \alpha_l \ket{l} \frac{1}{\sqrt{|H_l|}} \sum_{h\in H_l} \ket{h(12\dots n)}.
\end{equation}
In this section, we first design a $\poly(\log n)$-depth quantum algorithm for SIL symmetrization with perfect success probability in \cref{appendix:newsil}, which is used as a subroutine to solve the subgroup superposition preparation problem in \cref{appendix:exactnsil}.

\subsection{An exact SIL symmetrization algorithm}
\label{appendix:newsil}

Our exact SIL symmetrization algorithm relies on a classical parallel algorithm for generating a random permutation proposed by \cite{alonso1996parallel}. However, that algorithm was introduced in a non-reversible way. Therefore, in this subsection, we present that algorithm in a pedagogical manner and indicate along the way how to make it reversible.

\subsubsection{Lower exceeding sequences}
 The central tool used in the \textcite{alonso1996parallel} algorithm is the concept of lower exceeding sequences (LESs).
\begin{definition}[Lower exceeding sequences]
    A sequence $s\in \{1,2,\dots,n\}^n$ is a lower exceeding sequence if and only if $1\leq s_i \leq i$ for all $i\in [n]$.
\end{definition}
The number of LESs of length $n$ is $n!$, while the number of permutations of $n$ items is also $|S_n|=n!$. Reference \cite{alonso1996parallel} efficiently constructs a bijective mapping between these sets and uses it to generate a random permutation by generating an LES and then applying the mapping. Our quantum algorithm is based on a similar approach.

The mapping of \cite{alonso1996parallel} is defined as follows.
\begin{definition}[$s^\sigma$ as in \cite{alonso1996parallel}]
The LES $s^\sigma \in [n]^n$ corresponding to a permutation $\sigma\in S_n$ has $i$th element
\begin{equation*}
    s^\sigma_i = \left| \{ j : \sigma^{-1}(j) \leq \sigma^{-1}(i); j\leq i \} \right|
\end{equation*}
for all $i\in[n]$.
\end{definition}
It is easy to verify that $s^\sigma$ is an LES for any $\sigma\in S_n$, but it is less intuitive to explain the meaning of the values in $s^\sigma$. Fortunately, \textcite{alonso1996parallel} give an intuitive way to understand the meaning of $s^\sigma$ diagrammatically, as follows.

First, from each permutation $\sigma\in S_n$, one can compute $s^\sigma$. Consider a permutation $\sigma$ acting on the list of first $n$ positive integers, giving $\sigma(123\dots n)$. Let $\sigma(123\dots n)_i$ denote the 
$i$th element of $\sigma(123\dots n)$. One can draw the permutation diagrammatically on an $n\times n$ grid: if $\sigma(12\dots n)_i = j$, then block $(i,j)$ of the grid is filled, where $i$ is the row index and $j$ is the column index.
We call this a \emph{permutation-LES diagram}. The $j$th element of the generated LES is simply the number of filled blocks in the rectangle between $(1,1)$ and $(i,j)$. Since each column and row has and only has 1 filled block, the number of filled blocks between $(1,1)$ and $(i,j)$ is guaranteed to be at most $j$. \Cref{fig:les2permu} depicts an example with $n=6$.

\begin{figure}[t]
    \centering
\includegraphics[width=0.4\textwidth]{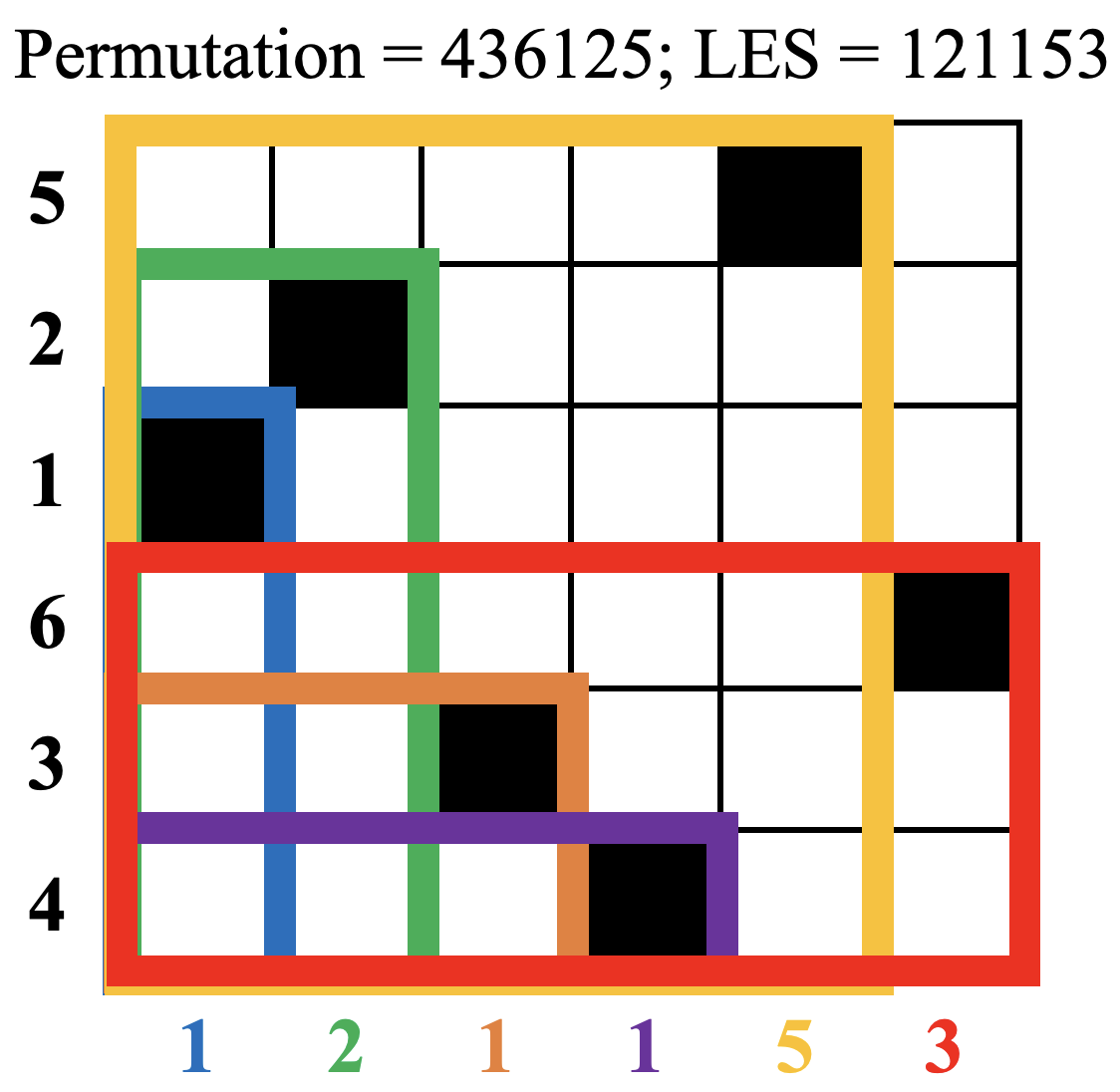}
\caption{The permutation-LES diagram for $\sigma(123456)=436125$, which appears on the vertical axis. The corresponding LES is $121153$, as shown on the horizontal axis. To determine the element of the LES in a given column from the permutation, we draw a rectangle whose bottom left corner is at the origin and whose top right corner is at the filled square in the corresponding column, and count the number of filled blocks in that rectangle. The rectangles are colored to indicate the corresponding element of the LES (i.e., the column).}
    \label{fig:les2permu}
\end{figure}

Conversely, from each LES, one can generate its corresponding permutation by again drawing the permutation-LES diagram. Suppose the LES is $(s_1 s_2 \dots s_n)$ and our goal is to compute $y_j$, the $y$-coordinate of the filled block in the $j$th column, for all $j \in \{1,\ldots,n\}$. We run through the LES from the last element $s_n$ to the first element $s_1$. To have exactly $s_n$ filled blocks between $(1,1)$ and $(y_n,n)$, we must fill the $s_n$th block of the last column, which means $y_n=s_n$. Since every row and column in the diagram should only have 1 filled block, we call the $y_n$th row \emph{used}.
Next, we observe that the $s_{n-1}$th unused row must be filled in the second-to-last column such that there are exactly $s_{n-1}$ filled blocks between $(1,1)$ and $(y_{n-1}, n-1)$, because all of the first $s_{n-1}-1$ unused rows will eventually be used and counted when determining the $(n-1)$st element of the LES. 
Similarly, we can find all other values of $y_j$ from $j=n-2$ to $j=1$ by filling the $s_{j}$th unused row of the $j$th column. Finally, the permutation can be read out directly from the diagram.

A naive implementation of the above procedure for converting an LES to its corresponding permutation takes $O(n\log n)$ depth. Fortunately, \cite{alonso1996parallel} provides a $\poly(\log n)$-depth algorithm to do the conversion, which is outlined next.

\subsubsection{From LES to permutation with reversible merging}

\textcite{alonso1996parallel} show that the blocks can be filled in a divide-and-conquer manner, as follows. Divide the LES into two parts at the middle, make the sub-diagrams for both parts individually, and merge them together. The first two steps are straightforward, but the ``merging" step is nontrivial. A simple merging strategy is not hard to see: one can insert every row of the left diagram into the unfilled rows of the right diagram. More specifically, we first concatenate two diagrams together and keep the right diagram unchanged. Then, if the $j$th column of the left diagram has the $y_j$th block filled, we move this block to the $y_j$th \emph{unfilled} row of the right diagram.

This merging strategy is illustrated in \cref{fig:merge} where the same permutation as in \cref{fig:les2permu} ($\sigma(123456)=436125$) is recovered. 

Performing this merging naively takes $O(n)$ steps for each pair of sub-diagrams. To reduce it to $O(\log n)$ steps, \textcite{alonso1996parallel} introduce another way of representing the diagrams: a sorted list of $(r_i,c_i)$ tuples, where each $(r_i,c_i)$ represents the row and column coordinates of a filled block in the diagram. We describe the general algorithm together with a running example that shows how it works in the case depicted in \cref{fig:les2permu}.

By default, the filled blocks are sorted by their row coordinates, i.e., the comparison rule $\Gamma_f$ is
\begin{equation}
(r_i,c_i) <_{\Gamma_f} (r_j,c_j) \quad \mathrm{iff} \quad r_i < r_j.    
\end{equation}
We can reduce the problem of merging diagrams to calculating the final $\Gamma_f$-sorted list, denoted by $T_f$, which contains all row indices from $1$ to $n$, i.e.,
\begin{equation}
    T_f = (1,z_1) (2,z_2) \dots (n,z_n).
\end{equation}
Once such a list is obtained, the corresponding permutation is simply $\sigma(12\dots n) = z_1 z_2 \dots z_n$. For instance, the $T_f$ corresponding to $\sigma(123456)=436125$ in \cref{fig:merge} is
\begin{equation}
    T_f = (1,4) (2,3) (3,6) (4,1) (5,2) (6,5).
\end{equation}

\begin{figure*}[t]
    \centering
    \includegraphics[width=0.88\textwidth]{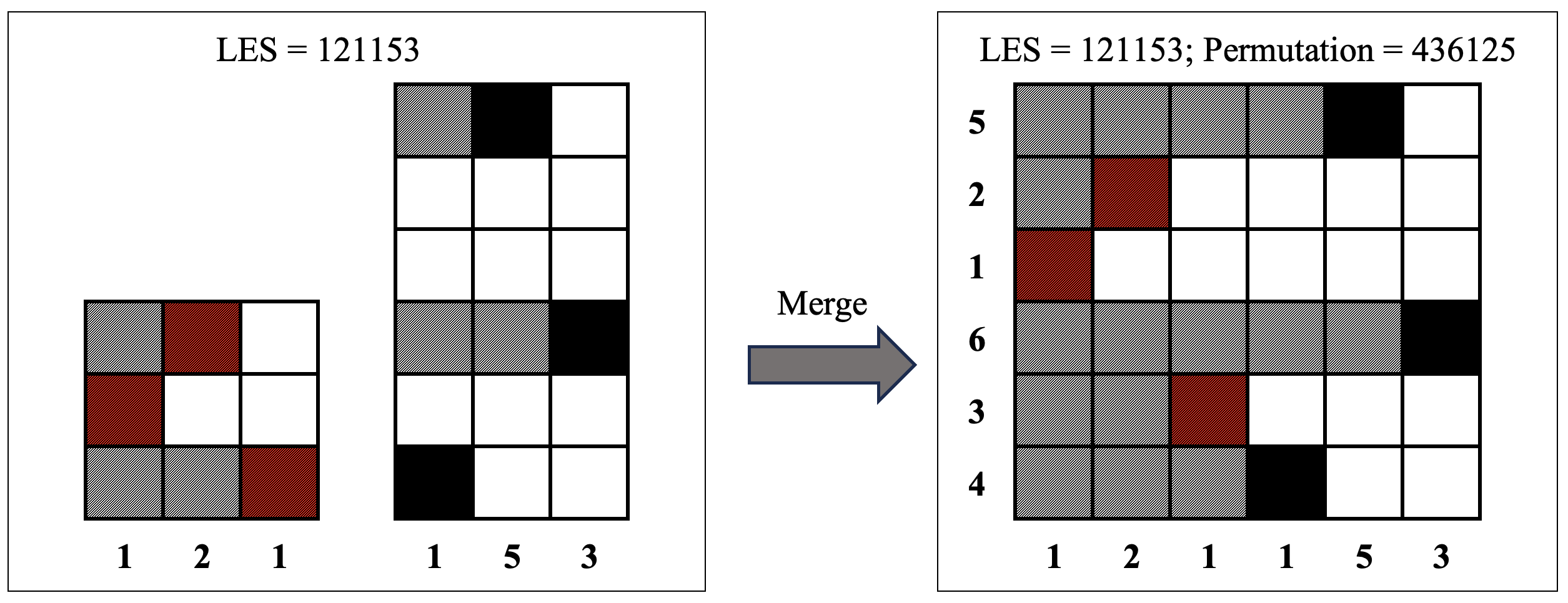}

    \caption{Action of the merging algorithm on the permutation $\sigma(123456)=436125$ (whose permutation-LES diagram is illustrated in \cref{fig:les2permu}). The gray squares are ``unavailable" since each row can only have one filled square. Dark red squares indicate filled blocks of the left diagram. When merging the two subdiagrams, the $i$th row in the left diagram is inserted into the $i$th unfilled row of the right diagram.}
    \label{fig:merge}
\end{figure*}

The key strategy of \textcite{alonso1996parallel} is to first calculate the $\Gamma_f$-sorted tuple lists for the left and right sub-diagrams, and merge them together to build the whole diagram. Note that each sub-diagram's list can also be calculated recursively using the same divide-and-conquer method. Without loss of generality, we only describe the last step of the algorithm, i.e., merging the last two sub-diagrams to obtain $T_f$, since all previous steps use the same procedure. We also emphasize that we describe the whole algorithm in a \emph{reversible} manner---unlike in the original paper \cite{alonso1996parallel}---so that it can be directly quantized.

We use $T_0$ to denote the direct concatenation of two $\Gamma_f$-sorted lists of both sub-diagrams. If both lists have $k$ elements, then
\begin{equation}
    T_0 = (r_1,c_1)(r_2,c_2)\dots (r_k,c_k) (r_{k+1},c_{k+1}) \dots (r_{2k},c_{2k}),
\end{equation}
where the first half and the second half are both $\Gamma_f$-sorted.
$T_0$ is sorted under the following comparison rule, denoted by $\Gamma_0$: $(r_i,c_i)<_{\Gamma_0} (r_j,c_j)$, if either of the following two conditions is satisfied: 
\begin{itemize}
	\item $c_i\leq k$ AND $c_j > k$;
	\item $r_i<r_j$ AND ($c_i,c_j\leq k$ OR $c_i,c_j > k$).
\end{itemize}

For the example in \cref{fig:merge}, $k=3$, and the $\Gamma_f$-sorted tuple list of the first subfigure is $(1,3)(2,1)(3,2)$ for the left-hand side and $(1,4)(3,6)(6,5)$ for the right-hand side. Therefore
\begin{equation}
	T_0 = (1,3)(2,1)(3,2)(1,4)(3,6)(6,5).
\end{equation}

Our objective is to convert $T_0$ to $T_f$ in a reversible manner. This can be done by first sorting $T_0$ reversibly such that every $c_i$ is at its correct position, i.e., we need to permute $T_0$ using some permutation $\omega$ such that $c_{\omega(i)} = z_i$. This sorted list, denoted by $T_1$, is 
\begin{equation}
\begin{aligned}
    T_1 =& (r_{\omega(1)},c_{\omega(1)}) (r_{\omega(2)},c_{\omega(2)}) \dots (r_{\omega(2k)},c_{\omega(2k)}) \\
    =& (r_{\omega(1)},z_1) (r_{\omega(2)},z_2) \dots (r_{\omega(2k)},z_{2k}).
\end{aligned}
\end{equation}
For the example in \cref{fig:merge},
\begin{equation}
T_1 = (1,4)(1,3)(3,6)(2,1)(3,2)(6,5).
\end{equation}
Note that the sorting procedure plays the same role as inserting elements from the left-hand side in the diagrams in \cref{fig:merge} to the right-hand side. In the second step, we change the $r_{\omega(i)}$ values in $T_1$ to $i$ reversibly so that $T_f$ is obtained.
We explain how to implement both steps as follows.

\textbf{Step 1: $T_0\mapsto T_1$.}
We reversibly sort $T_0$ to obtain $T_1$ using comparison-based sorting algorithms, based on the rules given by \textcite{alonso1996parallel}. 
First, we define $\tilde{T}_0$ and $\tilde{T}_1$ by concatenating each tuple in the lists with its index:
\begin{equation}
\begin{aligned}
    \tilde{T}_0 &= (1,r_1,c_1) (2,r_2,c_2) \dots (2k, r_{2k}, c_{2k}) \\
    \tilde{T}_1 &= (\omega(1),r_{\omega(1)},c_{\omega(1)}) (\omega(2),r_{\omega(2)},c_{\omega(2)}) 
    \dots\\
    &\quad\dots(\omega(2k),r_{\omega(2k)},c_{\omega(2k)}) \\
    &=  (\omega(1),r_{\omega(1)},z_1) (\omega(2),r_{\omega(2)},z_2)\dots\\
    &\quad \dots (\omega(2k),r_{\omega(2k)},z_{2k}).
\end{aligned}
\end{equation}
It is trivial to reversibly construct $\tilde{T}_0$ from $T_0$, and $\tilde{T}_0$ is also $\Gamma_0$-sorted.
For the example in \cref{fig:merge},
\begin{equation}
\begin{aligned}
    \tilde{T}_0 &= (1,1,3) (2,2,1) (3,3,2) (4,1,4) (5,3,6) (6,6,5)\\
    \tilde{T}_1 &= (4,1,4) (1,1,3) (5,3,6) (2,2,1) (3,3,2)(6,6,5).
\end{aligned}
\end{equation}

According to Ref.~\cite{alonso1996parallel}, $\tilde{T}_1$ is a $\Gamma_1$-sorted tuple list, where $\Gamma_1$ is defined as follows.
(Here we simply describe $\Gamma_1$ without proving why $\tilde{T}_1$ is $\Gamma_1$-sorted; see Ref.~\cite{alonso1996parallel} for the proof.) For $\omega(i) \neq \omega(j)$, $(\omega(i),r_{\omega(i)},c_{\omega(i)})<_{\Gamma_1} (\omega(j),r_{\omega(j)},c_{\omega(j)})$ if and only if any of the following conditions is satisfied:
\begin{enumerate}
    \item $\omega(i),\omega(j)\leq k$ AND $\omega(i)<\omega(j)$;
    \item $\omega(i),\omega(j) > k$ AND $\omega(i)<\omega(j)$;
    \item $\omega(i)\leq k, \omega(j)>k$ AND $r_{\omega(i)} + (\omega(j)-k) \leq r_{\omega(j)}$;
    \item $\omega(i)>k, \omega(j)\leq k$ AND $r_{\omega(j)} + (\omega(i)-k) \leq r_{\omega(i)}$.
\end{enumerate}

One can now get $\tilde{T}_1$ by sorting $\tilde{T}_0$ using any reversible sorting network without knowing $\omega$ explicitly, since $\tilde{T}_1 = \mathcal{A}_{\Gamma_1}(\tilde{T}_0)$ and $\tilde{T}_0 = \mathcal{A}_{\Gamma_0}(\tilde{T}_1)$. This step takes $\OT(\log n)$ depth.

Next, we convert $\tilde{T}_1$ to $T_1$ by removing the first element of each tuple, $\omega(i)$, reversibly. This requires us to compute the position of the tuple in the $\tilde{T}_0$ list using $i$, $r_{\omega(i)}$, and $c_{\omega(i)}$. A key observation is that, if $(i,r_i,c_i)<_{\Gamma_0} (j,r_j,c_j)$ and $c_i,c_j \leq k$, then $(\omega(i),r_{\omega(i)},c_{\omega(i)}) <_{\Gamma_1} (\omega(j),r_{\omega(j)},c_{\omega(j)})$, which also holds when $c_i,c_j > k$. This implies that the relative orders of each sub-diagram are preserved by sorting under $\Gamma_1$. Therefore, we can check whether $(\omega(i),r_{\omega(i)},c_{\omega(i)})$ is from the left diagram or the right diagram (which is easy by comparing $c_{\omega(i)}$ and $k$) and compute its rank in its sub-diagram (denoted by $d_{L,\omega(i)}$ if it is in the left sub-diagram or $d_{R,\omega(i)}$ if it is in the right sub-diagram) and compute
\begin{equation}
\omega(i) = 
\begin{cases}
    d_{L,\omega(i)}, & c_{\omega(i)}\leq k\\
    d_{R,\omega(i)} + k,& c_{\omega(i)} > k.
\end{cases}
\end{equation}
Now, it suffices to compute all $d_{L,\omega(i)}$ and $d_{R,\omega(i)}$ values. To do so, we first initialize both arrays by
\begin{equation}
\begin{array}{ll}
        d_{L,\omega(i)} \leftarrow 1,\, d_{R,\omega(i)} \leftarrow 0 &\mathrm{if}~ c_{\omega(i)} \leq k \\
        d_{L,\omega(i)} \leftarrow 0,\, d_{R,\omega(i)} \leftarrow 1 &\mathrm{if}~ c_{\omega(i)} > k.
\end{array}
\end{equation}
Then, using the prefix sum algorithm described in \cref{subsec:prefix}, both arrays can be computed in $O(\log^2 n)$ depth.

\textbf{Step 2: $T_1\mapsto T_f$.}
Converting $T_1$ to $T_f$ is essentially performing $r_{\omega(i)} \mapsto i$. Observe that if $c_{\omega(i)}>k$, then $r_{\omega(i)} = i$ and nothing needs to be done, since filled blocks in the right sub-diagram stay in the same rows during the merge. For $c_{\omega(i)} \leq k$, this update can be easily done by $r_{\omega(i)} \leftarrow r_{\omega(i)}+d_{R,\omega(i)}$ since a filled block in the left sub-diagram is moved to another row in the final diagram, and the new row is its row index in the left sub-diagram ($r_{\omega(i)}$) plus the number of rows below ($d_{R,\omega(i)}$) that are filled by blocks in the right sub-diagram. This process takes $O(\log^2 n)$ depth since it involves computing $d_{R,\omega(i)}$.

\subsubsection{The SIL symmetrization algorithm}
Since we can implement the merging process ($T_0 \mapsto T_1 \mapsto T_f$) in depth $\OT(\log^2 n)$ ($\OT(\log n)$ for $\tilde{T}_0 \mapsto \tilde{T}_1$, $\OT(\log^2 n)$ for $\tilde{T}_1 \mapsto T_1$, and $\OT(\log^2 n)$ for $T_1\mapsto T_f$), and this divide-and-conquer algorithm uses $O(\log n)$ layers of merging, the total depth of the LES-to-permutation conversion is $\OT(\log^3 n)$. To generate the superposition of all permutations of the nontrivial input $l\in\SIL$ with $|l|=n$ and $l\neq (123\dots n)$ in low depth, we ``move'' the permutations stored in $\sigma(12\dots n)$ to $l$ using the following quantum algorithm. Note that all ancilla qubits involved are omitted.
\begin{enumerate}
    \item Generate the equal superposition of all LESs of length $n$: $\ket{0}^{\otimes n} \mapsto \frac{1}{\sqrt{n!}} \bigotimes_{i=1}^n \sum_{s_i=1}^i \ket{s_i}$.
    \item Generate permutations from LESs using the quantized \textcite{alonso1996parallel} algorithm:\\
    $\frac{1}{\sqrt{n!}} \bigotimes_{i=1}^n \sum_{s_i=1}^i \ket{s_i} \mapsto \frac{1}{\sqrt{n!}} \sum_{\sigma\in S_n} \ket{\sigma (12\dots n)} $.
    \item Sort each permuted list using any (irreversible) sorting network:\\
    \begin{equation}
        \begin{aligned}
            &\frac{1}{\sqrt{n!}} \ket{00\dots0}_1 \sum_{\sigma\in S_n} \ket{\sigma (12\dots n)}_2\\
            \stackrel{\sort_{2,1}}{\longrightarrow}  &\frac{1}{\sqrt{n!}} \sum_{\sigma\in S_n} \ket{\rec(\sigma)}_1\ket{12\dots n}_2
        \end{aligned}
    \end{equation}
    where register 2 can be dropped since it is independent of $\sigma$.
    \item Unsort the input list $l$ using the records:
    \begin{equation}
    \begin{aligned}
        &\frac{1}{\sqrt{n!}} \sum_{\sigma\in S_n} \ket{\rec(\sigma)}_1 \ket{l}_2\\
        \stackrel{\unsort_{2,1}}{\longrightarrow} & \frac{1}{\sqrt{n!}}\sum_{\sigma\in S_n} \ket{00\dots 0}_1 \ket{\sigma(l)}_2
    \end{aligned}
    \end{equation}
    where register 1 can be dropped since it is independent of $\sigma(l)$.
\end{enumerate}

The whole process can be implemented as a single unitary, so it is guaranteed to produce the ideal symmetrized state.

\subsection{NSIL symmetrization for superposed inputs}
\label{appendix:exactnsil}
Recall that our algorithm for NSIL symmetrization involves preparing the equal superposition over the $H_l$ subgroup. 
Our strategy consists of two steps: detecting duplicates and preparing random permutations of duplicates in superposition.

\begin{algorithm}[H]
\caption{Subroutine $B_{h,t}(n_{h,1},n_{t,1},n_{h,2},n_{t,2})$}
        \label{alg:bht}
	\begin{algorithmic}
        \If{$h=t$}
        \State{$n_h\gets 1$}
        \State{$n_t\gets 1$}
        \Else
        \State{$h_1\gets h$}
        \State{$t_1\gets(h+t)/2$}
        \State{$h_2\gets (h+t)/2+1$}
	\State{$t_2\gets t$}
        \If{$l_{t_1}=l_{h_2}$}
            \If{$l_{t_1}\neq l_{h_1}$ AND $l_{h_2}\neq l_{t_2}$}
                \State $\dup(l)_{t_1 - n_{t,1}}\gets n_{t,1}+n_{h,2}$
                \State $n_h \gets n_{h,1}$
                \State $n_t \gets n_{t,2}$
            \ElsIf{$l_{t_1}\neq l_{h_1}$ AND $l_{h_2}=l_{t_2}$}
                \State $n_h \gets n_{h,1}$
                \State $n_t \gets n_{t,2}+n_{t,1}$
            \ElsIf{$l_{t_1} = l_{h_1}$ AND $l_{h_2} \neq l_{t_2}$}
                \State $n_h \gets n_{h,1}+n_{h,2}$
                \State $n_t \gets n_{t,2}$
            \ElsIf{$l_{t_1} = l_{h_1}$ AND $l_{h_2} = l_{t_2}$}
                \State $n_h \gets n_{h,1}+n_{h,2}$
                \State $n_t \gets n_{h,1}+n_{h,2}$
            \EndIf
        \Else
            \State $n_h \gets n_{h,1}$
            \State $n_t \gets n_{t,2}$
            \If{$1<n_{t,1}<(t-h)/2$}
                \State $\dup(l)_{t_1 - n_{t,1}} \gets n_{t,1}$
            \EndIf
            \If{$1<n_{h,2}<(t-h)/2$}
                \State $\dup(l)_{h_2} \gets n_{h,2}$
            \EndIf
        \EndIf
        \EndIf
        \end{algorithmic}
\end{algorithm}

\subsubsection{Duplicate detection algorithm}
We now describe an $\OT(\log n)$-depth quantum algorithm to detect all duplicated items in an ordered list. This algorithm takes as input $\ket{l}_1\ket{0}_2$ where $l$ is an NSIL and each integer in $l$ is assigned $O(\log n)$ qubits in register 2 to record how many times this integer is repeated in $l$. The output of this algorithm is $\ket{l}\ket{\dup(l)}$ where $\dup(l)$ records all repetitions: if there exists $n_i>1$ for $i$ such that $l_{i-1}<l_i=l_{i+1}=\dots=l_{i+n_i-1}<l_{i+n_i}$, then $\dup(l)_{i}=n_i$; for all other cases, $\dup(l)_{i}=0$. For example, if $l=(13333338)$, then $\dup(l)=(06000000)$.

We start by describing the classical algorithm to find $\dup(l)$, which can be quantized afterwards. The basic element of this algorithm is a ``block,'' denoted by $B$. Each block labeled by $h,t$ (``head'' and ``tail'') processes a range $[h,t]$ of $l$. The algorithm receives information from 2 blocks in the previous layer processing its sub-ranges $[h,(h+t)/2]$ and $[(h+t)/2+1,t]$, and checks whether there is a range of repeated elements across the 2 sub-ranges by checking whether $l[(h+t)/2] = l[(h+t)/2+1]$. Each block outputs 2 integers to the next layer, $n_h$ and $n_t$, which are the numbers of elements equal to $l_h$ and $l_t$, respectively, in the range $[h,t]$. The algorithm takes 4 integers from the previous layer, which are the output values of $B_{h,(h+t)/2}$, denoted by $n_{h,1}, n_{t,1}$, and of $B_{(h+t)/2+1,t}$, denoted by $n_{h,2},n_{t,2}$. A complete description of $B_{h,t}$ can be found in \cref{alg:bht}.

The whole algorithm contains $O(\log n)$ layers of blocks to find all ranges of duplicates bottom-to-top. A block in the $i$th layer of this algorithm processes a range of $2^{i-1}$ bits. To quantize the above duplicate detection algorithm, it suffices to quantize every block $B$ using a unitary $U_B$ with ancilla qubits.

\begin{figure}
    \centering
            \begin{quantikz}
                \lstick{$\ket{n_{h,1}}$} &\ctrl{4}&\qw &\\
                \lstick{$\ket{n_{t,1}}$} &\control{} &\qw &\\
                \lstick{$\ket{n_{h,2}}$} &\control{}&\qw&\\
                \lstick{$\ket{n_{t,2}}$} &\control{}&\qw &\\
                \lstick{$\ket{0}$} & \gate[3]{U_B} &\qw\rstick{$\ket{n_h}$} \\
                \lstick{$\ket{0}$} & &\qw\rstick{$\ket{n_t}$}\\
                \lstick{$\ket{0}$} & &\qw\rstick{$\ket{\dup(l)_{\frac{h+t}{2}-n_{t,1}}}$}\\
            \end{quantikz}
    \caption{Quantization of $B$.}
    \label{fig:enter-label}
\end{figure}

\subsubsection{Subgroup superposition preparation}
Now that the positions and sizes of repetitions have been recorded, we can modify the exact SIL symmetrization algorithm to prepare $\sum_{h\in H_l} \ket{h(12\dots n)}$ as follows. Note that $H_l$ only contains permutations of repeated items, and all other items must stay at their initial positions. We claim that $H_l$ corresponds to a specific class of LESs, denoted by $\LES_l$, which is a set of LESs of length $n$ satisfying the following conditions:
    \begin{itemize}
    \item if $l_{i-1}<l_i<l_{i+1}$, then $s_i = i$;
    \item if $l_{a-1}<l_a=l_{a+1}=\dots=l_{b}<l_{b+1}$, then $a \leq s_i \leq i$ for all $a\leq i \leq b$.
\end{itemize}
More formally,
\begin{equation}
    \LES_l := \{s : \forall i\in[n], \min_{l_j=l_i} j \leq s_i \leq \max_{l_j=l_i} j \}.
\end{equation}
\begin{theorem}
Each $h\in H_l$ can be generated from an $s\in \LES_l$, and each $s\in \LES_l$ can be computed from an $h\in H_l$, using the rules in \cref{appendix:newsil}.
\end{theorem}
\begin{proof}
    To show that there is a one-to-one mapping, we verify that $|H_l|=|\LES_l|$. Suppose $V_l$ is the set of all different elements in $l$, and each $v\in V_l$ appears $n_v$ times in $l$. Then
    \begin{equation}
        |H_l| = \prod_{v \in V_l} n_v!.
    \end{equation}
    Similarly, for any $v$ with $n_v=1$ and $l_i = v$, the corresponding $s \in \LES_l$ must satisfy $s_i = i$. If $n_v>1$, then there are $n_v$ $v$s in $l$ and there are $n_v!$ possible values allowed. Therefore, we also have
    \begin{equation}
        |\LES_l| = \prod_{v \in V_l} n_v! = |H_l|.
    \end{equation}

    One can use the rules in \cref{appendix:newsil} to calculate the corresponding $s^h\in \LES_l$ from an $h\in H_l$ as follows. Since permutations only exist between repeated elements, and repeated elements are adjacent in $l\in 
    \NSIL_n$, we observe that 
    \begin{equation}
        \min_{l_j=l_i} j \leq h(12\dots n)_i \leq\max_{l_j=l_i} j,
    \end{equation}
    and
    \begin{equation}
        h(12\dots n)_j < h(12\dots n)_i < h(12\dots n)_k
    \end{equation}
    if $l_j<l_i<l_k$. Recall that $s^h_{h(12\dots)_i}$ in the permutation-LES diagram of $h(12\dots n)$ is determined by counting the number of filled blocks from $(1,1)$ to $(h(12\dots)_i,i)$. The equation above implies that for all $l_j<l_i$, $h(12\dots)_j$ must be between $(1,1)$ and $(h(12\dots)_i,i)$, so 
    \begin{equation}
        s^h_i \geq 1 + \max_{l_j<l_i} j = \min_{l_j=l_i} j.
    \end{equation}
    For all $l_j=l_i$, since permutations are allowed, all those repeated items may be included in the range, so 
    \begin{equation}
        s^h_i \leq n_{l_i} + \max_{l_j<l_i} j = \max_{l_j=l_i} j.
    \end{equation}
    Therefore, $s^h \in \LES_l$. This immediately implies that one can use the rules in \cref{appendix:newsil} to find the corresponding $h\in H_l$ from $s^h$.
 \end{proof}

Therefore, to produce a superposition over $H_l$, it suffices to employ the SIL symmetrization algorithm with the LES state preparation step modified such that the initial state is
\begin{equation}
    \sum_{l\in \NSIL_n} \alpha_l \ket{l} \sum_{s\in \LES_l}\ket{s}.
\end{equation}
Since the pattern of repetitions in $l$ has been found, this state can still be prepared as a product state conditioned on the repetitions in $l$. The circuit depth is $O(1)$.

\section{Efficient conversion from second quantization to first quantization}
\label{sec:binaryencoder}
To convert any given second-quantized state to its first-quantized counterpart, we first transform the second-quantized state to an NSIL before employing the (anti-)symmetrization algorithm on the increasing lists. To do so, we design a novel reversible circuit that takes as input a list of occupation numbers (i.e., a second-quantized representation) and outputs an increasing list of positions (or modes) occupied, where each mode is assigned with a number written in its binary representation. More specifically, our circuit transforms $W_0:= n_0 n_1 \dots n_{m-1}$ (where $n_i\geq 0$ is the occupation number of mode $i$) into $W_f := 0^{n_0} 1^{n_1} \dots (m-1)^{n_{m-1}}$. Our approach is based on reversible sorting networks and takes $\OT(\log^2 n)$ depth where $n=\sum_{i=0}^{m-1} n_i$.
The main idea for converting $W_0$ to $W_f$ reversibly is to first generate $n$ empty registers; then use sorting networks to move $n_i$ of them in between of the $(i-1)$th register and the $i$th register, for all $i\in[m]$; and use the new positions of the empty registers to compute the index $i$ locally for each of them. In more detail, the algorithm proceeds as follows.

First, we calculate the prefix sums of $W_0$ and produce the list of $2$-tuples:
\begin{equation}
    W_1 = (n_0,n_0) (n_1,n_0+n_1) (n_2,n_0+n_1+n_2) \dots (n_{m-1},n).
\end{equation}
This can be done in $O(\log^2 n)$ depth as described in \cref{subsec:prefix}. We also add indices to all elements in $W_1$, giving a $3$-tuple list
\begin{equation}
\begin{aligned}
    \tilde{W}_1 &= (0,n_0,n_0) (1,n_1,n_0+n_1) (2,n_2,n_0+n_1+n_2)\\ 
    &\quad\dots (m-1,n_{m-1},n).
\end{aligned}
\end{equation}

Next, since the output list contains $n$ integers, we add another $n$ elements with $(m,0,j)$ where $j\in\{0,1,\dots,n-1\}$ is the index of the new element added. The resulting list is
\begin{equation}
\begin{aligned}
    \tilde{W}_2 &= (0,n_0,n_0) (1,n_1,n_0+n_1) \dots\\
    &\quad (m-1,n_{m-1},n) (m,0,0) (m,0,1) \dots (m,0,n-1).
\end{aligned}
\end{equation}
Note that $\tilde{W}_2$ is $\Gamma_{W_2}$-sorted where $(a_1,b_1,c_1) <_{\Gamma_{W_2}} (a_2,b_2,c_2)$ if and only if either of the following 2 conditions is satisfied:
\begin{itemize}
    \item $a_1<a_2$;
    \item $a_1=a_2=m$ AND $c_1<c_2$.
\end{itemize}

Now, using reversible sorting networks, we sort $\tilde{W}_2$ to obtain
\begin{equation}
\begin{aligned}
    \tilde{W}_3 &= (m,0,0)(m,0,1)\dots (m,0,n_0-1) (0,n_0,n_0) (m,0,n_0)\\
    &\quad \dots (m,0,n_0+n_1-1) (1,n_1,n_0+n_1) (m,0,n_0+n_1) \\ 
    &\quad \dots (m,0,n-1) (m-1,n_{m-1},n),
\end{aligned}
\end{equation}
which is $\Gamma_{W_3}$-sorted where $(a_1,b_1,c_1)<_{\Gamma_{W_3}} (a_2,b_2,c_2)$ if and only if either of the following 2 conditions is satisfied:
\begin{itemize}
    \item $c_1<c_2$,
    \item $c_1 = c_2$ AND $a_1<a_2$.
\end{itemize}

Now, we have allocated $n_i$ such elements (with first element being $m$) between $(i-1,n_{i-1},\sum_{j=0}^{i-1}n_j)$ and $(i,n_i,\sum_{j=0}^{i}n_j)$, and we need to assign the value $i$ to each of them. For the $k$th element in $\tilde{W}_3$, if the first value is $m$ (meaning that this is a newly added element) and the third value is $i_k$ (the index of new elements), then the corresponding $i$ value is simply $k-i_k$ (the difference between the index in $\tilde{W}_3$ and the index in all new elements). We add this value to the second value for each new element, obtaining
\begin{equation}
\begin{aligned}
    \tilde{W}_4 &= (m,0,0)(m,0,1)\dots (m,0,n_0-1) (0,n_0,n_0) \\
    &\quad (m,1,n_0) \dots (m,1,n_0+n_1-1) (1,n_1,n_0+n_1) \\ 
    &\quad (m,2,n_0+n_1) \dots (m,m-1,n-1) \\
    &\quad (m-1,n_{m-1},n).
\end{aligned}
\end{equation}
$\tilde{W}_4$ is also $\Gamma_{W_3}$-sorted since we only modify the second values, which do not play a role in the $\Gamma_{W_3}$ comparison rule.
Now, if we only consider elements whose first value is $m$, we have already obtained the desired $W_f$ list in their second values. Hence, in the final stage, we remove all other elements reversibly. We first sort $\tilde{W}_4$ under $\Gamma_{W_2}$ to group together elements with first values less than $m$, giving
\begin{equation}
\begin{aligned}
    \tilde{W}_5 &= (0,n_0,n_0) (1,n_1,n_0+n_1)\\
    &\quad\dots (m-1,n_{m-1},n) (m,0,0) \dots.
\end{aligned}
\end{equation}
We can first remove all second values for the first $m$ elements by subtracting the difference between $\sum_{j=0}^i n_j$ and $\sum_{j=0}^{i-1}n_j$, which is equal to $n_i$. For $i=0$, it is trivial to implement $(0,n_0,n_0)\mapsto(0,0,n_0)$. This gives us
\begin{equation}
    \tilde{W}_6 = (0,0,n_0) (1,0,n_0+n_1) \dots (m-1,0,n) (m,0,0)\dots,
\end{equation}
which is still $\Gamma_{W_2}$-sorted. We sort it again under $\Gamma_{W_3}$ to obtain
\begin{equation}
\begin{aligned}
    \tilde{W}_7 &= (m,0,0)\dots (m,0,n_0-1) (0,0,n_0) (m,1,n_0)\dots\\
    &\quad (m,1,n_0+n_1-1) (1,0,n_0+n_1) (m,2,n_0+n_1) \\ 
    &\quad\dots (m,m-1,n-1) (m-1,0,n).
\end{aligned}
\end{equation}
Next, we notice that $(i,0,\sum_{j=0}^i n_i)$ must be the $(\sum_{j=0}^i n_i +i)$th element in $\tilde{W}_7$ for $i\in\{0,1,m-1\}$. We can therefore subtract all $\sum_{j=0}^i n_i$ values and obtain
\begin{equation}
\begin{aligned}
    \tilde{W}_8 &= (m,0,0)\dots (m,0,n_0-1) (0,0,0) (m,1,n_0)\dots\\
    &\quad (m,1,n_0+n_1-1) (1,0,0) (m,2,n_0+n_1) \dots \\
    &\quad (m,m-1,n-1) (m-1,0,0).
\end{aligned}
\end{equation}
Next, we can sort $\tilde{W}_8$ again under $\Gamma_{W_2}$ to obtain
\begin{equation}
\begin{aligned}
    \tilde{W}_f &= (0,0,0)\dots(m-1,0,0) (m,0,0)\dots \\
    &\quad (m,0,n_0-1) (m,1,n_0)\dots(m,1,n_0+n_1-1)\dots\\
    &\quad (m,m-1,n-1).
\end{aligned}
\end{equation}
We now have a list where the first $m$ elements are just counting up to $m$, the first values of all remaining elements are $m$, and the third values of the remaining elements are all just the indices. All of this unnecessary information can be reversibly removed, yielding the desired output list
\begin{equation}
    W_f = 0^{n_0} 1^{n_1} \dots (m-1)^{n_{m-1}}.
\end{equation}
Since all operations in this process are either reversible sorting or reversible local data modification, the whole process including symmetrization takes $\OT(\log^3 n)$ depth and can be implemented as a reversible circuit.

Due to the reversibility of this algorithm, we can also transform a (superposition) of NSILs back to a second-quantized state by inverting the above procedure. Since our NSIL symmetrization algorithm is also reversible, combing both algorithms gives a mechanism to efficiently convert between first quantization and second quantization.

\section{Application: quantum interferometric imaging system}
\label{sec:inter}
\subsection{QFT-based single-photon interferometry}
\label{subsec:singlephoton}
We review quantum telescope arrays \cite{gottesman2012longer,khabiboulline2019quantum} for the low-photon-rate scenario with the help of the schematic diagram in \cref{fig:telescope}. For simplicity, we work in an idealized model of a 2-dimensional world where all $m$ telescopes are equally spaced in a line. Each telescope $T_0,T_1,\dots,T_{m-1}$ has a lens (or mirror) to collect photons and a well-isolated photon detector equipped with a bosonic quantum memory (since photons are bosons).
The distance between adjacent telescopes is $d$. The photon source (which may be a star, galaxy, etc.) is denoted by $S$ and is assumed to be infinitely far away from the telescope array, so its position can be described using a single parameter $\theta$, the angle between $T_0 S$ and $SP$ where $P$ is the nearest point to $S$ on the telescope plane. Since both $d$ and $m$ are finite, the following angle relation can be derived from trigonometry:
\begin{equation}
    \theta = \angle T_0 T_1 M = \angle T_0 S P = \angle T_1 S P =\dots = \angle T_{m-1} S P.
\end{equation}

\begin{figure}[t]
    \centering
    \includegraphics[width=0.9\linewidth]{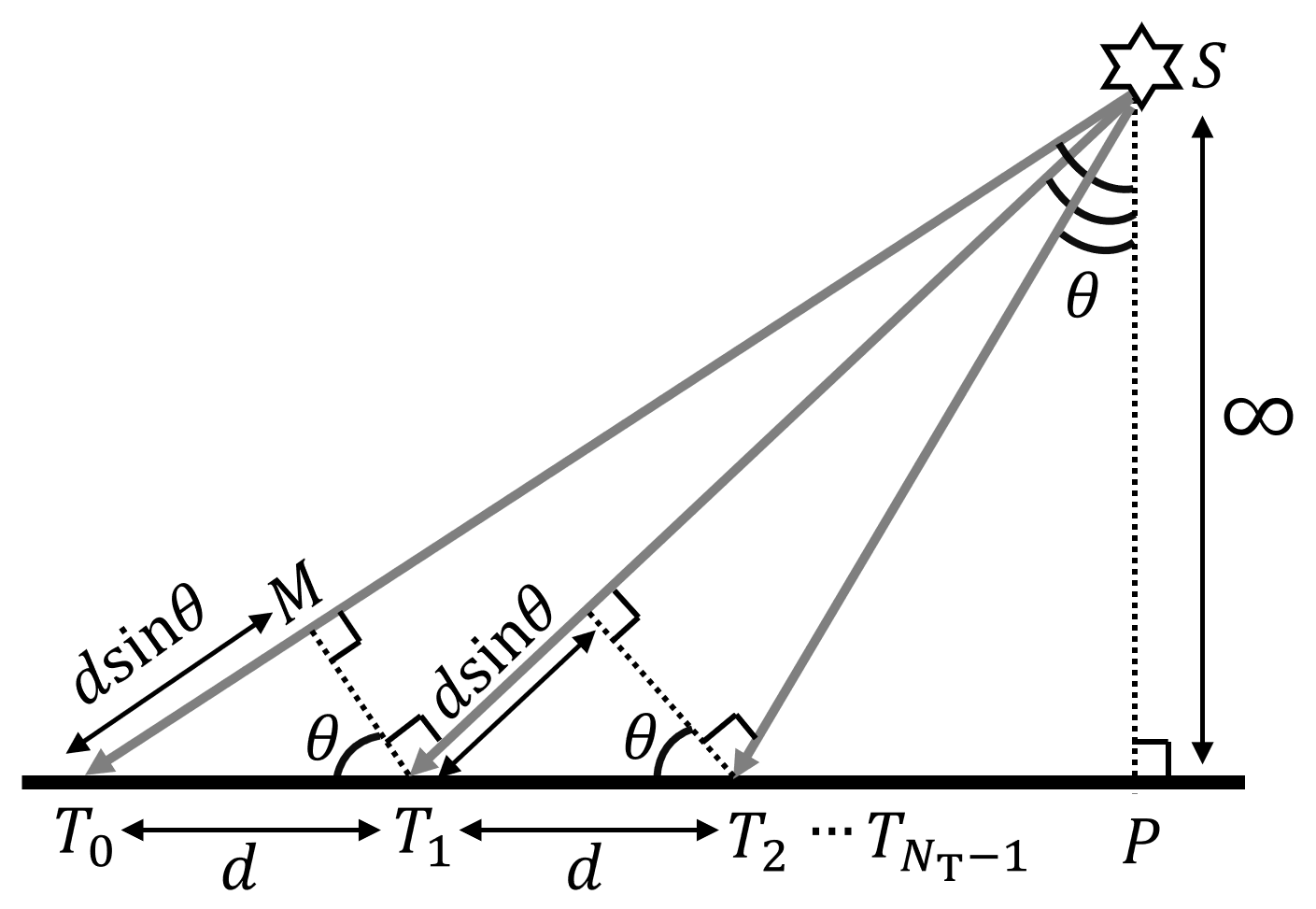}
    \caption{In this idealized 2-D model, since $S$ is infinitely far away from the telescope array, $\theta = \angle T_0 T_1 M = \angle T_0 S P = \angle T_1 S P $. Therefore, the path difference between $ST_j$ and $ST_{j+1}$ is $d\sin\theta$.}
    \label{fig:telescope}
\end{figure}

A photon emitted by $S$ may finally be received by one of the $m$ detectors, but the path length gone through by the photon varies for different detectors. For two adjacent detectors in the array, the path difference is $d\sin\theta$ by trigonometry, as shown in \cref{fig:telescope}. The path difference leads to a phase difference, $2\pi d\sin\theta / \lambda$. If we let the phase at detector $T_0$ be 0, then the phase at detector $T_j$ is $2\pi dj \sin\theta / \lambda$. Since the photon may arrive at any detector in the array, we can define the creation operator of the array $a^\dagger(\theta)$ with angle $\theta$ in terms of the creation operator of the memory of the $j$th detector, $b^\dagger_j$, for all $j\in\{0,1,\dots,m-1\}$: 
\begin{equation}
    a^\dagger(\theta) = \sum_{j=0}^{m-1} \exp(\imag \frac{2\pi  dj \sin\theta}{\lambda}) b^\dagger_{j}.
\end{equation}
Note that $[b^\dagger_i,b^\dagger_j]=0$, $[b_i,b^\dagger_j]=\delta_{ij}$, and $b^\dagger_i \ket{v_i}_i = \sqrt{v_i+1}\ket{v_i+1}_i$ since we are using a bosonic memory.

The initial state of the quantum memories is $\ket{0}_0\ket{0}_1\dots \ket{0}_{m-1}$, also written as $\ket{00\dots 0}$ for short. The state after a single photon detection is
\begin{equation}
\label{eq:single_photon_state}
\begin{aligned}
    &a^\dagger(\theta)\ket{00\dots 0} \\
    =&\frac{1}{\sqrt{m}} \sum_{j=0}^{m-1} e^{\imag \frac{2\pi  dj \sin\theta}{\lambda}} \ket{0}^{\otimes j-1} \ket{1}_j \ket{0}^{m-j},
\end{aligned}
\end{equation}
where the state is represented using second quantization. This is a natural representation because the photon detectors and their corresponding quantum memories are far from each other and the occupation numbers can be stored locally.

For most optical-wavelength astronomical observations, the number of photons received per time step is much less than $1$. In this case, we need only be able to process the above state. Note that the only information we need from this state is the angle $\theta$, since an image is essentially a photon distribution over angles. We notice that there is a periodicity in the phases depending on the $\theta$, so the QFT is a natural tool to extract the information. To perform a QFT, we can transform the state into a first-quantized binary representation:
\begin{equation}
    \ket{0}^{\otimes j-1} \ket{1}_j \ket{0}^{m-j} \mapsto \ket{j}
\end{equation}
where $j$ is represented by up to $\log(m)$ qubits. To do so, we first bring the distributed quantum memories together, which can be done by either physically moving the qubits
or using quantum repeaters. Next, we run a standard unary-to-binary encoder in $\log(m)$ depth. The transformed state is simply
\begin{equation}
    \ket{\psi(\theta)} = \frac{1}{\sqrt{m}} \sum_{j=0}^{m-1} \exp(\imag \frac{2\pi  dj \sin\theta}{\lambda}) \ket{j}.
\end{equation}
We can now perform the QFT on this state to obtain
\begin{equation}
    \frac{1}{m} \sum_{j=0}^{m-1} \sum_{k=0}^{m-1} \exp(2\pi \imag j \left( \frac{ d\sin\theta}{\lambda} -\frac{ k}{m}\right)) \ket{k}.
\end{equation}
Observe that if 
\begin{equation}
     \frac{ d\sin\theta}{\lambda} -\frac{ k}{m} = v
\end{equation}
where $v$ can be any integer, then the Fourier-transformed state is simply $\ket{k}$, where
\begin{equation}
\label{eq:kdefinition}
    k \left( \frac{\lambda}{m d} \right) = \sin\theta - \frac{v\lambda}{d}.
\end{equation}
The physical meaning of $k$ is now straightforward: if the telescope array only receives photons from $\sin\theta = \frac{v\lambda}{d}$ to $\sin\theta = \frac{(v+1)\lambda}{d}$ (which can be enforced by the lens/mirror in front of the photon detector), then $k$ is the difference between $\sin\theta$ and $\frac{v\lambda}{d}$ in terms of $\frac{\lambda}{m d}$. Here $m d$ is the \emph{baseline}, $\frac{\lambda}{d}$ is the \emph{field of view}, and $\frac{\lambda}{m d}$ is the \emph{angular resolution}, which saturates the Rayleigh criterion. Since it is not necessary to physically bring photons together for interference, this method can be applied to longer baselines and shorter wavelengths (infrared or optical) than classical interferometry, potentially leading to micro-arcsecond angular resolution. Furthermore, since this telescope array can handle single photon-events, for photon rate $\epsilon$, it has sensitivity $O(\epsilon)$, much higher than the $O(\epsilon^2)$ sensitivity of intensity interferometers. 

Note that in the discussions of quantum telescope arrays in this paper, we assume there are no atmospheric fluctuations that may produce local phase differences. This may be reasonable if the telescope array is space-based, but for ground-based telescopes, such phases are inevitable and may affect image quality, since the phase difference between the $i$th and the $j$th telescopes is no longer strictly $2\pi d (j-i) \sin\theta /\lambda$. These fluctuations can potentially be corrected using classical methods such as adaptive optics \cite{beckers1993adaptive}, as mentioned in Ref.~\cite{khabiboulline2019opticalprl}.

\subsection{Multiple-photon interferometry}
\label{subsec:multiplephoton}
So far, we have only discussed the case where the detected state has the form $a^\dagger(\theta)\ket{00\dots 0}$. More generally, it is natural to consider multiple-excitation states of the form
\begin{equation}
    \ket{\psi(\theta_1,\dots,\theta_{n})} = \prod_{j=1}^{n} a^\dagger(\theta_j) \ket{00\dots 0},
\end{equation}
where $n$ is the number of photons received by the telescope array in one time step, and $\theta_j$ is the angle of the $j$th photon.

This multiple-photon scenario is not relevant to the most common situations in optical or infrared interferometry in astronomy. This is mainly because the signal from a distant astronomical object is very weak when it arrives at the telescope array. It is also related to the observed wavelengths. Typically, astronomical sources emit photons of all wavelengths. At radio or microwave wavelengths, since the energy per photon ($hc/\lambda$) is low, the number of photons per mode can be more than 1. However, very long baseline interferometry has been realized in these wavelengths, making quantum interferometry unnecessary. The more astronomically interesting regime is optical and infrared wavelengths, for which we design quantum telescopes arrays to perform interferometry. However, the photon rate in this case is usually much lower than 1 per mode due to the much higher energy per photon.

However, we claim that the photon rate can be much higher when one or more of the following conditions hold: 
\begin{enumerate}[noitemsep]
    \item the number of detectors $m$ is large;
    \item the area of each aperture is large;
    \item the field of view is large (which can be achieved when $d$ is very small);
    \item the photon wavelength is long; or
    \item the source is bright.
\end{enumerate}
These conditions might be met in more general imaging tasks beyond observational astronomy, so we call a quantum ``telescope" array working under these conditions an \emph{interferometric imaging system}. Note that even under these conditions, we could be in a regime where the number of photons is much less than the number of modes, making the first-quantized representation advantageous.

Now, we outline how our NSIL symmetrization algorithm makes it possible to extract all $\theta_j$ from $\ket{\psi(\theta_1,\dots,\theta_{n})}$ using only one copy of the state using \cref{result:2}.
Let
\begin{equation}
    \phi_{i,j} := \frac{2\pi dj \sin\theta_i}{\lambda}.
\end{equation}
The explicit multiple-photon state is
\begin{equation}
\begin{aligned}
    &\prod_{j=1}^{n} a^\dagger(\theta_j) \ket{00\dots 0}\\
    &\propto \sum_{r\in [m]^{n}} \exp(\imag \sum_{i=1}^{n} \phi_{i,r_i}) b_{r_1}^\dagger b_{r_2}^\dagger \dots b_{r_{n}}^\dagger \ket{00\dots 0},
\end{aligned}
\end{equation}
where $r$ is a list of integers whose $i$th element $r_i$ represents at which telescope the $i$th photon is detected. Now, letting $n_i$ be the number of photons detected at telescope $i$ in $r$,
\begin{equation}
\label{eq:detected_state}
    b_{r_1}^\dagger b_{r_2}^\dagger \dots b_{r_{n}}^\dagger \ket{00\dots 0} = \sqrt{\prod_{i=0}^{m-1} n_i!} \ket{n_0 n_1 \dots n_{m-1}}
\end{equation}
due to the bosonic nature of the quantum memory. This is the second-quantized representation of the photon state.

As with the single-photon case, we convert to the first-quantized representation in order to apply the quantum Fourier transform.
To do so, as described in \cref{sec:binaryencoder}, we first transform $\ket{n_0 n_1 \dots n_{m-1}}$ into a list $l$ of telescope indices in increasing order. This transformation circuit has $\OT(\log^2 n)$ depth and uses $O(n\log m)$ ancilla qubits.
As a simple example, if $m=3$, $n=4$, and $r=2101$, then the state in Equation~\eqref{eq:detected_state} is $\ket{121}$, which will be transformed to $\ket{l(r)} = \ket{0112}$. Here $l(r)=0112$ is an NSIL as it is the sorted version of $r$. We also notice that the factor $\sqrt{\prod_{i=0}^{m-1} n_i!}$
is just $\sqrt{|H_{l(r)}|}$, square root of the order of the subgroup of $S_{n}$ preserving $l(r)$. We can thus write the transformed unnormalized state as
\begin{equation}
\begin{aligned}
    \sum_{r\in [m]^{n}} \exp(\imag \sum_{i=1}^{n} \phi_{i,r_i}) \sqrt{\left| H_{l(r)} \right|}\ket{l(r)}. 
\end{aligned}
\end{equation}

Next, we can apply the NSIL symmetrization algorithm in \cref{result:2} to obtain the unnormalized first-quantized state
\begin{equation}
\label{eq:psis}
\begin{aligned}
&\ket{\psi_\mathrm{S}(\theta_1,\dots,\theta_{n})}\\
    &= \sum_{r\in [m]^{n}} \exp(\imag \sum_{i=1}^{n} \phi_{i,r_i}) \sqrt{\frac{\left| H_{l(r)} \right|}{\left| S_{l(r)} \right|}} \sum_{\sigma'\in S_{l(r)}} \ket{\sigma'(l(r))}\\
    &\propto \sum_{r\in [m]^{n}} \exp(\imag \sum_{i=1}^{n} \phi_{i,r_i}) \left| H_{l(r)} \right| \sum_{\sigma'\in S_{l(r)}} \ket{\sigma'(l(r))},
\end{aligned}
\end{equation}
where we use $\left| S_{n} \right| = \left| S_{l(r)} \right| \cdot \left| H_{l(r)} \right| $ in the last line.

To make it easier to analyze this state, we rewrite it as a superposition over $r \in [m]^{n}$, i.e.,
\begin{equation}
    \ket{\psi_\mathrm{S}(\theta_1,\dots,\theta_{n})} \propto \sum_{r\in [m]^{n}} \beta_{r} \ket{r}
\end{equation}
where $\beta_r$ is the coefficient corresponding to $\ket{r}$. We notice that if $r'$ satisfies $l(r') = l(r)$, then its phase in Equation~\eqref{eq:psis} contributes to $\beta_{r}$, and it can be related to $r'$ by $r' = \sigma'(r)$ where $\sigma' \in S_{r}$. Therefore,
\begin{equation}
\begin{aligned}
    \beta_{r} &= \left| H_{l(r)} \right|  \sum_{\sigma \in S_{r}}  \exp(\imag \sum_{i=1}^{n} \phi_{i,\sigma'(r)_i})\\
    &= \sum_{\sigma \in S_{n}}  \exp(\imag \sum_{i=1}^{n} \phi_{i,\sigma(r)_i})
\end{aligned}
\end{equation}
where in the second line we replace the summation over $S_{r}$ by that over $S_{n}$ and use the degeneracy in $S_{n}$ when acting on $r$, i.e., $|H_{l(r)}|\cdot |S_{r}| = |S_{n}|$. We also observe that the sum is preserved under a permutation of photon indices, so
\begin{equation}
\begin{aligned}
    \beta_r &= \sum_{\sigma \in S_{n}}  \exp(\imag \sum_{i=1}^{n} \phi_{\sigma^{-1}(12\dots n)_i,\sigma^{-1}\circ \sigma(r)_i})\\
    &= \sum_{\sigma \in S_{n}}  \exp(\imag \sum_{i=1}^{n} \phi_{\sigma(12\dots n)_i,r_i}),
\end{aligned}
\end{equation}
where in the last line we replace $\phi_{\sigma^{-1}(12\dots n)_i,r_i}$ by $\phi_{\sigma(12\dots n)_i,r_i}$ because $\sigma^{-1}$ is unique for each $\sigma$ and $\sigma^{-1}\in S_{n}$.

Finally, we rewrite
\begin{equation}
\begin{aligned}
    &\ket{\psi_\mathrm{S}(\theta_1,\dots,\theta_{n})}\\
    &\propto \sum_{r\in [m]^{n}} \sum_{\sigma \in S_{n}}  \exp(\imag \sum_{i=1}^{n} \phi_{\sigma(12\dots n)_i,r_i})\ket{r}\\
    &= \sum_{\sigma \in S_{ n}} \sum_{r\in [m]^{ n}}   \exp(\imag \sum_{i=1}^{ n} \phi_{\sigma(12\dots n)_i,r_i}) \ket{r}\\
    &= \sum_{\sigma \in S_{ n}} \bigotimes_{i=1}^{ n} \left( \sum_{r_i\in[m]} \exp(\imag \phi_{\sigma(12\dots n)_i,r_i}) \ket{r_i} \right).
\end{aligned}
\end{equation}
All $\theta_i$ can be extracted by running the QFT on all $ n$ registers in parallel and measuring in the computational basis:
\begin{equation}
\mathrm{QFT}^{ n} \ket{\psi_\mathrm{S}(\theta_1,\dots,\theta_{ n})} = \frac{1}{\sqrt{n!}} \sum_{\sigma \in S_{ n}}  \bigotimes_{i=1}^{ n} \ket{k_{\sigma(12\dots  n)_i}}
\end{equation}
where $k_{\sigma(12\dots  n)_i} = \frac{md}{\lambda} \left(\sin\theta_{\sigma(12\dots  n)_i} - \frac{v\lambda}{d}\right)$ with $v\in \mathbb{Z}$ determined by the direction of the telescope, as in Equation~\eqref{eq:kdefinition}.

Thus, with the help of the NSIL symmetrization algorithm, a single-photon quantum telescope array can be upgraded to an interferometric imaging system capable of processing multiple photons at a time. The quantum computation for $n$ photons on $m$ detectors has depth $\poly(\log n)$ and uses $O(m\log n + n\log m)$ qubits. Capturing multiple photons may allow imaging systems to work for more general purposes, including high-resolution observation of the sun and astronomical surveys. Our scheme may also be particularly useful in some specific wavelengths, such as the far infrared, where the photon rate is close to $1$ but classical very long baseline interferometry is not practical.

Another potential application in the far future is photography in dark conditions.
Note that to achieve field of view $\frac{\lambda}{d}$ with angular resolution $\frac{\lambda}{Nd}$, it suffices to have an $N\times N$ array of low-resolution lenses/mirrors with aperture radii much shorter than $d$, instead of a single lens/mirror with aperture $Nd$. Manufacturing large-aperture lenses is difficult and expensive even for ordinary photography, which can be avoided with our imaging system at the cost of a quantum information processor and well-isolated photon detectors.

It is also worth mentioning that instead of directly producing images, quantum telescope arrays can offer better accuracy in estimating specific parameters, such as the distance between binary stars, as introduced in Ref.~\cite{sajjad2024quantum}. Our multiple-photon interferometric imaging system can potentially extend the list of high-precision parameter estimations.

\section{Summary \& Discussion}
\label{sec:summary}
In this paper, we studied the general quantum symmetrization problem. We identified subtleties with repetitions in the input list and pointed out issues with previous SIL symmetrization algorithms when acting on NSILs.
Using sorting networks, we then proposed the first logarithmic-depth quantum algorithms for symmetrizing NSILs. These algorithms leverage the coset decomposition of the symmetric group, use an ancilla register to mediate multiplications of permutations elements, and adapt SIL symmetrization algorithms as subroutines. The single-input symmetrization algorithm features $\OT(\log n)$ depth and uses $O(n\log n + \log m)$ ancilla qubits for NSILs of $n$ integers with value up to $m$, while the superposed-input symmetrization algorithm uses $\OT(\log^3 n)$ depth and $O(n\log n)$ ancilla qubits.
We also proposed an $\OT(\log^2 n)$-depth algorithm to transform between second-quantized states and $\NSIL$s, which allows for efficient conversion between second-quantized states and first-quantized states.

Our algorithms finally enable first-quantized simulation of bosons by giving a mechanism for preparing symmetrized initial states. The NSIL symmetrization algorithm can be used to prepare (superpositions of) Dicke states of any Hamming weight in logarithmic depth using a reasonable number of ancilla qubits. The efficient conversion algorithm between first quantization and second quantization provides a mechanism for processing multiple photons at a time in the quantum telescope array proposal.

We distinguish between the scenario in our paper and in other algorithmic applications of symmetrization. It is well known that the graph isomorphism problem can be solved by creating a superposition over all permutations of a given graph \cite{ambainis2011symmetry}. However, our results cannot be used in that context since there is no increasing input guarantee in the setting of graph isomorphism. Indeed, if the vertices of a graph could be ordered in a canonical way, this would solve graph isomorphism.

Our algorithm requires all-to-all connectivity due to the sorting networks used, just as in Ref.~\cite{berry2018improved}. A natural open problem is to find low-depth sorting networks with restricted connectivity, which could make our algorithm much more favorable for near-term devices.
It might also be possible to reduce the circuit depth of the superposed-input symmetrization algorithm in \cref{result:2} or increase the success probability of the algorithm in \cref{result:1}.

\section*{Acknowledgments}
We thank Alexey Gorshkov, Emil Khabiboulline, Cheng-Ju Lin, and Yuan Su for helpful discussions.
We acknowledge support from the National Science Foundation (QLCI grant OMA-2120757).

\bibliography{ref}

\end{document}